\documentclass{IEEEtran}
\onecolumn
\usepackage[pdftex]{graphicx}
\usepackage{amsmath}
\usepackage{amssymb}
\usepackage{amsthm}
\usepackage{mathtools}

\usepackage{cancel}
\usepackage{tikz}

\newtheorem{claim}{Claim}
\newtheorem{theorem}{Theorem}
\newtheorem{corollary}{Corollary}
 
\newtheorem{definition}{Definition}
\newtheorem{remark}{Remark}
\newtheorem{lemma}{Lemma }

\newtheorem{notation}{Notation}

\newcommand{\wt}{\widetilde}
\newcommand{\wh}{\widehat}

\newcommand{\mse}{MSE}
\newcommand{\mb}{\mathbf}
\newcommand{\mbt}{\mathbf}

\newcommand{\lra}{\leftrightarrow}

\newcommand{\ol}{\overline}
\newcommand{\ch}{\ensuremath{\mathrm{conv}}}
\usepackage{amsmath}
\usepackage{amssymb}
\usepackage{graphicx}
\usepackage{latexsym}
\usepackage{cite}
\usepackage{extraipa}
\usepackage{mathrsfs}
\usepackage{stmaryrd}
\usepackage{overpic}
\usepackage{pdfsync}

\newcommand{\set}[1]{\mathscr{#1}}

\begin{document}

\sloppy

\title{ LP Bounds for Rate-Distortion with Variable Side Information} 
\author{Sinem Unal and Aaron B. Wagner~
\thanks{S.~Unal was with  Cornell University, School of Electrical \& Computer Engineering, Ithaca, NY 14853 USA. She is now with KenCast Inc. Norwalk, CT 06854 USA.  A.~B.~Wagner is with  Cornell University, School of Electrical \& Computer Engineering, Ithaca, NY 14853 USA (e-mail: su62@cornell.edu, wagner@ece.cornell.edu). This paper was presented in part at the IEEE Int. Symposium on Information Theory (ISIT), Barcelona, July 2016 and it is submitted for presentation in part to Data Compression Conference (DCC), Snowbird, UT April, 2017.}
}

\maketitle

\begin{abstract}
We consider a rate-distortion problem with side information at multiple decoders. Several upper and lower bounds have been proposed for this general problem  
or special cases of it.
We provide an upper bound for general instances of this problem, 
which takes the form of a linear program, by utilizing random binning and simultaneous decoding techniques \cite{ElGamal} and compare it with the existing bounds.
We also provide a lower bound for the general problem,
which was inspired by a linear-programming lower bound for index coding,
 and show that it subsumes most of the lower bounds in literature. 
Using these upper and lower bounds, we explicitly characterize the rate-distortion function of a problem that can be seen as a Gaussian analogue of the ``odd-cycle'' index coding problem.
\end{abstract}

\section{Introduction}

We consider the rate-distortion tradeoff for a canonical problem in source coding: an encoder with access to a source of interest broadcasts a single message to multiple decoders, each endowed with side information about the source. Each decoder then wants
to reproduce the source subject to a distortion constraint. This is essentially the multiple-decoder extension of the \textit{Wyner-Ziv}~\cite{wyner_ziv} 
problem, sometimes referred to as the \textit{Heegard-Berger}~\cite{berger}
problem.

Even for the two-decoder case, the complete characterization of the rate-distortion 
function is a long-standing open problem. However, the rate-distortion function has been determined in several special cases, including when the side information at the various
decoders can be ordered according to stochastic degradedness~\cite{berger},
when there are two decoders whose side information is ``mismatch degraded''
\cite{watanabe}, and when there are two decoders and the side information at decoder $2$ is ``conditionally less noisy'' than the side information at decoder 1 and decoder 1 seeks to 
losslessly reproduce a deterministic function of the source~\cite{timo_lessnoisy}. Also, instead of imposing some degraded structure on the side information, one can consider degraded reconstruction sets at the two decoders in which one component of the source is reconstructed at both decoders with vanishing block error probability and the other component of the source is only reconstructed at a single decoder \cite{benammar}. Various vector Gaussian instances of the problem  are solved \cite{sinem_gauss, sinem_CISS}. Several instances of the index coding problem, which is an important special case, have also been solved~(e.g.,~\cite{birk, blasiak,sinem_index}). 

Upper and lower bounds on the rate-distortion function in the
general case are also
available. Existing achievable schemes proceed by crafting
separate messages for different subsets of decoders, which
are encoded and decoded in a fixed order using random 
binning~\cite{berger,watanabe,timo}.
Our first contribution is to show  how such schemes can 
be improved using \emph{simultaneous decoding}~\cite{ElGamal},
in which each decoder decodes all of its messages at once
instead of sequentially. The resulting achievable bound 
involves optimizing over auxiliary random variables and,
for each choice of such variables, solving a linear program (LP).
Prior to this work, the best achievable bound was due to
Timo \textit{et al.}~\cite{timo}, who corrected an earlier achievable
bound
due to Heegard and Berger~\cite{berger}. In fact, as we 
discuss in Section \ref{sec:subsumes}, the proof given by
Timo \textit{et al.} contains an error similar to the one
contained in Heegard and Berger.

One natural way of obtaining a lower bound is to consider a relaxed instance of the problem in which the side information at some of the decoders is enhanced in such a way that the problem becomes stochastically degraded. Indeed, most existing lower bounds adopt this approach in some form~\cite{sinem_index,sinem_CISS}. For the special case of index coding, Blasiak~\emph{et al.}~\cite{blasiak} provide a lower bound that takes the form of a linear program, the constraints for which are derived from properties
of the entropy functional, such as submodularity. This  raises the question of whether a similar-style bound can be obtained for more general instances of the problem.
The second main contribution of the paper is such a bound. It is obtained by introducing a notion of \textit{generalized side information} and capturing the properties of mutual information in the form of a linear program. We show that this lower bound subsumes several existing lower bounds. 

To demonstrate the efficacy of our upper and lower bounds, we consider a rate-distortion problem obtained by extending the odd-cycle index coding problem to Gaussian sources with mean squared error (MSE) distortion constraints.  We find an explicit expression for its rate-distortion function by combining the two bounds.

The outline of the paper is as follows. Section \ref{sec:problem_description} formulates the general rate-distortion problem. Section \ref{sec:upper_bound} presents the  LP-type upper bound based on simultaneous decoding  while Section \ref{sec:upper_gauss_ext} provides the extension of this upper bound to  Gaussian sources. In Section \ref{sec:lower_bound}, we provide the LP-type lower bound and in Section \ref{sec:subsumes} we show that the LP-type upper and lower bounds subsume several existing bounds. Lastly Section \ref{sec:odd_cycle_gauss} presents optimality results including the rate distortion characterization of the  odd-cycle Gaussian problem.

\section{Problem Description}
\label{sec:problem_description}
Let $X$ denote the source at the encoder and $\mathcal{X}$ denote the source alphabet. Also, $Y_l \in \mathcal{Y}_l$, $l \in [m]$ denotes the side information at decoder $l$ and $Y_l$ is jointly distributed with the source, $X$. Lastly, $\wh{X}_{l}\in \wh{\mathcal{X}}_l$ denotes the reconstruction of $X$ at decoder $l$ and $D_l$ denotes the corresponding distortion constraint. Each decoder wishes to reconstruct the source, $X$, subject to its distortion constraint and we assume initially that the source alphabet, $\mathcal{X}$,  the side information alphabets, $\mathcal{Y}_l$, $l \in [m]$, and the reconstruction alphabets $\wh{\mathcal{X}}_l$, $l \in [m]$, are finite.  
We use the following definitions to formulate the problem.
\begin{definition}
An $(n, M, \mb{D})$ code where $n$ denotes the blocklength and $M$ denotes the message size and $\mb{D} =(D_1,\ldots, D_m )$ is composed of 
\begin{itemize}
\item an encoding function
\begin{align*}
f  : \mathcal{X}^{n}  \rightarrow  \{ 1, ... , M \}
\end{align*}
\item and decoding functions
\begin{align*}
g_1 & : \{ 1, ... , M \}  \times  \mathcal{Y}_1^{n}  \rightarrow \wh{\mathcal{X}}^n_1
\\
\vdots
\\
g_m & : \{ 1, ... , M \}  \times  \mathcal{Y}_m^{n}  \rightarrow \wh{\mathcal{X}}^n_m 
\end{align*}
\end{itemize}
satisfying the distortion constraints
\begin{align*}
 E\left[\frac{1}{n} 
\sum_{k = 1}^{n} d_l(X_k,\wh{X}_{lk})\right] \le D_l, \mbox{ for } l \in [m]
\end{align*}
where 
\begin{align*}
& \wh{X}^n_{l} = g_l(f(X^n), Y_l^n), \mbox{ for } l \in [m]
\end{align*}
 and $d_l(\cdot,\cdot) \in [0, \infty)$ is the distortion measure for decoder $l$.
 \end{definition}
\begin{definition}
A rate $R$ is \emph{$\mb{D}$-achievable} if for every $\epsilon > 0$ there exists an $(n,M, \mb{D} + \epsilon \mb{1})$ (where $\mb{1}$ is the all-ones vector) \emph{code} such that for sufficiently large $n$ we have $n^{-1}\log{M} \le R + \epsilon$. 
\end{definition}

We define the \emph{rate-distortion function} as
\begin{align*}
R(\mb{D}) = \inf\{R : R  \textrm{ is } \mb{D} \textrm{-achievable} \}.
\end{align*}


\section{Simultaneous Decoding Based Upper Bound to $R(\mb{D})$}
\label{sec:upper_bound}
We present our first main result, which is an upper bound to the rate-distortion function $R(\mb{D})$.  The following  notation, which is  similar to that in \cite{timo}, will be useful to state the results.
\begin{notation}
Let $(X,Y,Z)$ be a random vector. Then $X \perp Y$ denotes that $X$ and $Y$ are independent, $X \perp Y | Z$ denotes that $X$ and $Y$ are conditionally independent given $Z$, and $X \lra Y \lra Z$ denotes that $X$, $Y$ and $Z$ form a Markov chain.
\end{notation}

\begin{notation}
$v = \set{S}_1, \ldots, \set{S}_{2^m-1}$ denotes an ordered list of all nonempty subsets of $[m]$, where each $\set{S}_i$ denotes a different subset. $\set{V}$ denotes the set of all possible such $v$. 
\end{notation}
\begin{notation}
Let $v \in \set{V}$ be fixed. $\mathcal{U}_{\set{S}_1}$,$\ldots$, $\mathcal{U}_{\set{S}_{2^m-1}}$ denote the alphabets of finite-alphabet random variables $U_{\set{S}_1}$,$\ldots$, $U_{\set{S}_{2^m-1}}$ respectively. $\set{P}_v$ denotes the set of all distributions on $\mathcal{U}^*_v \times \mathcal{X} \times \mathcal{Y}^*$ where  $\mathcal{U}_v^* = \mathcal{U}_{\set{S}_1}$ $\times$ $\cdots$ $\times$ $\mathcal{U}_{\set{S}_{2^m-1}}$ and $\mathcal{Y}^* = \mathcal{Y}_{1}$ $\times$ $\cdots$ $\times$ $\mathcal{Y}_{m}$.
\end{notation}
\begin{notation}
\label{not3:mesg_sets}
Let $\set{U} = \{U_{\set{S}_1},U_{\set{S}_2},\ldots, U_{\set{S}_{2^m-1}}\}$,  $\set{D}_l = \{ \set{S}_i | \mbox{ } l \in \set{S}_i \}$, and $\set{D}'_l$ be a nonempty subset of  $\set{D}_l$. Then we define
\begin{align*}
    U_{\set{D}'_l} &= \left\{ U_{\set{S}_i} \in \set{U} |\mbox{ } \set{S}_i \in  \set{D}'_l \right\}
    \\
   U^{-}_{\set{S}_j} &= \Big\{U_{\set{S}_i} \in \set{U}\ |\ i < j \Big\}
 \\
      U^{-}_{\set{S}_j,\set{D}'_l} &= \left\{ U_{\set{S}_i} \in U^{-}_{\set{S}_j}  |\mbox{ } \set{S}_i \in  \set{D}'_l \right\}.
  \end{align*}
\end{notation}

\begin{theorem}
\label{thm:gen_ach}
The rate-distortion function, $R(\mb{D})$, is upper bounded by 
\begin{align}
&R_{ach}(\mb{D}) = \ch\left( R'_{ach}(\mb{D}) \right)
\end{align}
where $\ch(R'_{ach}(\mb{D}))$ denotes the lower convex envelope of $R'_{ach}(\mb{D})$ with respect to $\mb{D}$,
\begin{align}
\label{eq:gen_ach}
R'_{ach}(\mb{D}) = \min_{v \in \set{V}} \inf_{C_{ach, v}(\mb{D})} \inf_{C^{LP}_{ach}}  \sum^{2^m-1}_{j = 1} R_{\set{S}_j}, 
\end{align}
\begin{align*}
 C_{ach, v}(\mb{D}) &\ \text{denotes the set of } \  p \in \set{P}_v \mbox{ such that }
 \\
 & 1)  p(x,y_1,\ldots, y_m)  \mbox{ equals the } \mbox{joint distribution of } (X, Y_1,\ldots, Y_m)
 \\
& 2) \set{U} \lra X \lra (Y_1, \ldots, Y_m) 
\\
& 3) \mbox{ There exist functions } g_l (U_{\set{D}_l}, Y_l)
\mbox{ such that } E\left[d_l (X, g_l (U_{\set{D}_l}, Y_l))\right] \le D_l  \mbox{ for all } l \in [m],
\end{align*}
and
\begin{align*}
 C^{LP}_{ach} & \ \text{denotes the set of } \ R_{\set{S}_j}, R'_{\set{S}_j}, \mbox{ where } \set{S}_j \in v, \mbox{ such that } 
 \\
  & 1) R_{\set{S}_j} \ge 0,   R'_{\set{S}_j} \ge 0  \mbox{ for all } j \in [2^m-1] 
\notag \\
& 2) R_{\set{S}_j} \ge I\big(X,U^-_{\set{S}_j};U_{\set{S}_j}\big) - R'_{\set{S}_j}
\mbox{ for all } j \in [2^m-1]
\notag \\
& 3) \mbox{ For each decoder } l, l \in [m],
\notag \\
&\sum_{\set{S}_j \in \set{D}'_l} R'_{\set{S}_j} \le
\sum_{\set{S}_j \in \set{D}'_l}H(U_{\set{S}_j})  - H(U_{\set{D}'_l} | U_{\set{D}_l \setminus \set{D}'_l }, Y_l), \mbox{ for all } \set{D}'_l  \subseteq \set{D}_l. 
 \notag
\end{align*}
\end{theorem}
\begin{proof}[Proof of Theorem \ref{thm:gen_ach}]
The proof is given in the Appendix \ref{app:gen_ach}.
\end{proof}

\begin{remark}
Using the chain rule, we can rewrite condition 3) of $C^{LP}_{ach}$ in Theorem \ref{thm:gen_ach} as 
\begin{align}
\label{thm:ach_gen_LPconds_rep}
&\mbox{ for each decoder } l, l \in [m],
\notag \\
&\sum_{\set{S}_j \in\set{D}'_l} R'_{\set{S}_j} \le
\sum_{\set{S}_j \in \set{D}'_l} I\left( U_{\set{S}_j};U^{-}_{\set{S}_j,\set{D}'_l}, U_{\set{D}_l \setminus \set{D}'_l }, Y_l\right), \mbox{ for all } \set{D}'_l  \subseteq \set{D}_l. 
\end{align}
This representation will be useful when we extend this theorem to continuous sources. Hence, from now on we consider the condition 3) of  $C^{LP}_{ach}$ in the form of (\ref{thm:ach_gen_LPconds_rep}).
\end{remark}
\begin{remark}
Since $R_{ach}(\mb{D})$ is an upper bound to the rate-distortion function, $R(\mb{D})$,
we can obtain  a computable upper bound to $R(\mb{D})$ by imposing a cardinality constraint on the alphabets of auxiliary random variables $U_{\set{S}_j}$ in Theorem \ref{thm:gen_ach}.
\end{remark}
The scheme can be described as follows. Each $U_{\set{S}_j}$ in Theorem \ref{thm:gen_ach} can be viewed as a message for all decoders $l$ for which $l \in \set{S}_j$. The encoder encodes each message $U_{\set{S}_j}$ with respect to the order $v \in \set{V}$, using random binning.
Here $R_{\set{S}_j}$ and  $R'_{\set{S}_j}$ can be interpreted as the number of bins in the codebook of message $U_{\set{S}_j}$ and the number of codewords per bin, respectively.
Then each decoder $l$ decodes its messages using simultaneous decoding and reconstructs the source using these messages and its side information, $Y_l$.

\subsection{Rate-Distortion Function with Gaussian Source and Side Information}
\label{sec:upper_gauss_ext}
We next extend the achievable scheme in Theorem \ref{thm:gen_ach} to the rate-distortion problem with vector Gaussian sources.
More specifically, we are interested in the following rate-distortion problem.
The source and the side information at the decoders, $(\mb{X}, \mb{Y_1}, \ldots, \mb{Y_m})$, are zero mean jointly Gaussian vectors. The source $\mb{X} = (X_1,\ldots,X_k)$ has length $k$ and the length of $\mb{Y_i}$ is $k_i$, $i\in [m]$.
\begin{notation}
Let $\mb{v}$ and $\mb{w}$ be $k\times 1$ vectors. The notation $\mb{v} \le \mb{w}$ denotes that the $i^{th}$ component of $\mb{v}$, denoted by $v_i$, is less than or equal to $w_i$ for all $i \in [k]$. 
\end{notation}

\begin{notation}
Let $M$ be an $m\times m$ matrix. $(M)_d$ denotes the vector whose $i^{th}$
component is the $i^{th}$ diagonal element of $M$, $i \in [m]$.
\end{notation}

\begin{notation}
$K_{\mb{X}}$ denotes the covariance matrix of $\mb{X}$. $K_{\mb{X}|\mb{Y}}$ denotes the conditional covariance matrix of  $\mb{X}$ conditioned  on $\mb{Y}$.
\end{notation}
 Let $\mb{D_i} > 0 \mbox{ for all } i \in [m]$. The distortion constraints are 
\begin{align}
 \left(\frac{1}{n} 
\sum_{k = 1}^{n} E\left[(\mb{X}_{k} - \mb{\wh{X}}_{\mb{i}k})(\mb{X}_{k} - \mb{\wh{X}}_{\mb{i}k})^T\right] \right)_d \le \mb{D_i}, \mbox{ for all } i \in [m],
\end{align}
i.e., we impose component-wise mean square error ($\mse$) distortion constraints.
Since we have $\mse$ distortion constraints, without loss of generality we can take the reconstruction at each decoder to be the conditional expectation of the source given the output of the encoder and the corresponding side information. 
From now on, we denote the rate-distortion function of this problem as $R^{\mse}(\mb{D})$. 

\begin{theorem}
\label{thm:gen_ach_gauss}
Let the joint distribution of $(\mb{X},\mb{Y_i})$, $i \in [m]$ be given. Then the rate-distortion function, $R^{\mse}(\mb{D})$, is upper bounded by  
\begin{align}
\label{eq:gen_ach_gauss}
&R^G_{ach}(\mb{D}) = \min_{v \in \set{V}} R^G_{ach,v}(\mb{D})
\end{align}
where 
\begin{align*}
R^G_{ach,v}(\mb{D}) = \inf_{C^G_{ach,v}(\mb{D})} \inf_{C^{LP}_{ach}}  \sum^{2^m-1}_{j = 1} R_{\set{S}_j} 
\end{align*}
\begin{align*}
 C^G_{ach,v}(\mb{D}) & \ \text{denotes the set of }  \ p \in \set{P}_{v}  \mbox{ such that }
 \\
 & 1) p(\mb{x}, \mb{y_1}, \ldots, \mb{y_m})  \mbox{ equals the } \mbox{joint distribution of } (\mb{X}, \mb{Y_1}, \ldots, \mb{Y_m})
 \\
 &2) T \mbox{ is  a discrete random variable over $[\tau]$ for some positive integer $\tau$ such that $T \perp (\mb{X}, \mb{Y_1},\ldots,\mb{Y_m})$}
 \\
& 3) \set{U} \lra \mb{X} \lra (\mb{Y_1}, \ldots, \mb{Y_m}) 
\\
& 4)  K_{\mb{X}|\mb{U}_{\set{D}_i}, \mb{Y_i} } \le \mb{D_i}  \mbox{ for all } i \in [m]
\\
&5) \mb{U}_{\set{S}_j} = (\mb{U}_{\set{S}_j,t}, T)  \mbox{ such that }  \mb{U}_{\set{S}_j} = \mb{U}_{\set{S}_j,t} \mbox{ if } T=t, \mbox{ all } \mb{U}_{\set{S}_j,t} \mbox{ are jointly Gaussian with } (\mb{X},\mb{Y_1},\ldots,\mb{Y_m}),
\\
&\quad   \mbox{ and } I(\mb{U}_{\set{S}_j}; \mb{U}_{\set{S}_i},\mb{X}) < \infty
\mbox{ for all } \set{S}_j  \in [2^m-1], \set{S}_i \in [2^m-1] \mbox{ and } i \neq j,
\end{align*}
and $C^{LP}_{ach}$ is  the set of conditions obtained by  replacing each $X$, $Y_i$, and $U_{\set{S}_j}$ in the conditions of $C^{LP}_{ach}$ in Theorem \ref{thm:gen_ach} by  $\mb{X}$, $\mb{Y_i}$, and $\mb{U}_{\set{S}_j}$ respectively.
\end{theorem}
\begin{remark}
Since all feasible messages $\mb{U}_{\set{S}_j}$ in (\ref{eq:gen_ach_gauss}) are Gaussian mixtures  and the source and the side information vectors are jointly Gaussian, the minimum mean square error (MMSE) estimator is linear given the realization of $T$. In other words, we can write $\mb{\hat{X}_i} = A_{i,t}\mb{U}_{\set{D}_i} + B_{i,t}\mb{Y_i}$ given $T=t$, where the value of $A_{i,t}$ and $B_{i,t}$ are determined by the joint distribution $p  \in C^G_{ach,v}(\mb{D})$.
\end{remark}
\begin{proof}[Proof of Theorem \ref{thm:gen_ach_gauss}]
The argument is based on a quantization of the source and messages similar to the procedure in \cite[Section 3]{ElGamal}. First we quantize the source, all messages and the side information. Then we apply the achievable scheme in the proof of Theorem \ref{thm:gen_ach} to these quantized variables and show that the rate in (\ref{eq:gen_ach_gauss}) is $\mb{D}$-achievable for our problem.

Let $v \in \set{V}$ be fixed and $\epsilon > 0$ be given. Also let $(\mb{X}, \set{U}, \mb{Y_1},\ldots, \mb{Y_m})$ be such that the joint distribution of it, denoted by $p$, is in $C^G_{ach,v}(\mb{D})$.  Note that we can 
represent each message $\mb{U}_{\set{S}_j} = A_{\set{S}_j,t}\mb{X} + \mb{N}_{\set{S}_j,t}$, $\set{S}_j \in v$, if $T= t$ where $\mb{N}_{\set{S}_j,t} \perp (\mb{Y_1},\ldots,\mb{Y_m})$ and we can represent the side information as $\mb{Y_i} = B_{i}\mb{X} + \mb{N}_{i}$, $i \in [m]$ where $\mb{N}_{i} \perp (\set{U}, \mb{X})$.
Now we quantize $\mb{X}$ and all $\mb{Y_i}$, $i \in [m]$, and we use the notation $\mb{\ol{X}}$ to denote the quantized version of $\mb{X}$. We perform the quantization such that 

\begin{align}
\label{ineq:proof_ach_gen_gauss_quant1}
& E\left[(X_j - \ol{X}_j)^2 \right] \le \delta(\epsilon)\min_{i \in [m]} D_{ij} \mbox{ for all } j \in [k]
\\
\label{ineq:proof_ach_gen_gauss_quant2}
&E\left[(\ol{X}_j - \ol{\wh{X}}_{ij})^2\right] \le D_{ij} +\delta(\epsilon)D_{ij} \mbox{ for all } i \in [m] \mbox{ and } j \in [k]
\end{align}
\begin{align}
\label{ineq:proof_ach_gen_gauss_quant3}
& | I\big(\mb{X},\mb{U}^-_{\set{S}_j};\mb{U}_{\set{S}_j}\big) - I\big(\mb{\ol{X}},\ol{\mb{U}}^-_{\set{S}_j}; \mb{\ol{U}}_{\set{S}_j}\big) | \le \delta(\epsilon), \mbox{ for all } \set{S}_j \in [2^m-1]
\\
\label{ineq:proof_ach_gen_gauss_quant4}
&\left| \sum_{\set{S}_j \in \set{D}'_i} I\left( \mb{U}_{\set{S}_j};\mb{U}^{-}_{\set{S}_j,\set{D}'_i},\mb{U}_{\set{D}_i \setminus \set{D}'_i }, \mb{Y_i}\right) 
-
 \sum_{\set{S}_j \in \mb{\set{D}'_i}} I\left( \mb{\ol{U}}_{\set{S}_j};\ol{\mb{U}}^{-}_{\set{S}_j,\set{D}'_i},\ol{\mb{U}}_{\set{D}_i \setminus \set{D}'_i }, \mb{\ol{Y}_i}\right) \right| \le  \delta(\epsilon), \mbox{ for all } i\in [m] \mbox{ and  }\set{D}'_i \subseteq \set{D}_i,
\end{align}
where $\delta(\epsilon) > 0$ is to be specified later, 
and
\begin{align*}
\ol{\set{U}} \lra \mb{\ol{X}} \lra \mb{X} \lra (\mb{Y_1}, \ldots, \mb{Y_m}) \lra (\mb{\ol{Y}_1},\ldots, \mb{\ol{Y}_m}).
\end{align*}

Let  $\ol{p}$ denote the joint distribution of $(\ol{\set{U}}, \mb{\ol{X}}, \mb{\ol{Y}_1},\ldots, \mb{\ol{Y}_m})$. Now we form a new problem in which the source is $\mb{\ol{X}}$, the side information at decoder $i$ is 
$\mb{\ol{Y}_i}$, $i \in [m]$, and the distortion constraints are as in (\ref{ineq:proof_ach_gen_gauss_quant2}).
Note that for this problem, $\ol{p}$ is in  $C_{ach,v}((1+\delta(\epsilon))\mb{D})$ in (\ref{eq:gen_ach}). Then we can apply the achievable scheme in the proof Theorem \ref{thm:gen_ach} to the new problem.

Let  $R_{LP}((1+\delta(\epsilon))\mb{D},\ol{p})$ denote the result of the linear program $\inf_{C^{LP}_{ach}}  \sum^{2^m-1}_{j = 1} R_{\set{S}_j}$ in Theorem \ref{thm:gen_ach} when the joint distribution is  $\ol{p}$. Then from    Theorem $1$,  rate $R_{LP}((1+\delta(\epsilon))\mb{D},\ol{p})$ is $(1+\delta(\epsilon))\mb{D}$-achievable for the new problem. In other words, we can find an  $(n,M, (1+\delta(\epsilon))\mb{D} + \epsilon'\mb{1})$, $\epsilon'(\epsilon) >0$ (specified later), code with rate 
\begin{align}
\label{ineq:proof_qrate_req}
R_{LP}\left((1+\delta(\epsilon))\mb{D},\ol{p}\right) +\epsilon'(\epsilon)
\end{align}
and
 \begin{align}
 \label{ineq:proof_qdist_req}
  \left(\frac{1}{n} 
\sum_{j = 1}^{n} E\left[(\ol{\mb{X}}_{j} - \ol{\mb{\wh{X}}}_{\mb{i}j})(\ol{\mb{X}}_{j} - \ol{\mb{\wh{X}}}_{\mb{i}j})^T\right] \right)_d \le (1+\delta(\epsilon))\mb{D_i} + \epsilon'\mb{1} \mbox{ for all } i \in [m]
 \end{align}
when the blocklength, $n$, is sufficiently large. 

For our original problem, first we quantize the source, the side information and all the messages distributed by $p$ as described above and then we apply the $(n, M, (1+\delta(\epsilon))\mb{D} + \epsilon'\mb{1})$ code with rate (\ref{ineq:proof_qrate_req}) to these quantized variables, the joint distribution of which is $\bar{p}$.
Let $R^G_{LP}(\mb{D},p)$ denote the result of the linear program $\inf_{C^{LP}_{ach}}  \sum^{2^m-1}_{j = 1} R_{\set{S}_j}$ in Theorem \ref{thm:gen_ach_gauss} when the joint distribution is $p$. 
Note that the linear programs defining both  $R_{LP}\left((1+\delta(\epsilon))\mb{D},\ol{p}\right)$ and  $R^G_{LP}(\mb{D},p)$ are finite. Thus by (\ref{ineq:proof_ach_gen_gauss_quant3}), (\ref{ineq:proof_ach_gen_gauss_quant4}) and standard results on the continuity of linear programs \cite{cont_LP}, we have that 
\begin{align*}
|R^G_{LP}(\mb{D},p) - R_{LP}\left((1+\delta(\epsilon))\mb{D},\ol{p}\right)| \le \gamma(\epsilon),
\end{align*}
where $\gamma(\epsilon) \rightarrow 0$ as $\delta(\epsilon) \rightarrow 0$.
Lastly utilizing the Cauchy and Jensen inequalities and using (\ref{ineq:proof_ach_gen_gauss_quant1}) and (\ref{ineq:proof_qdist_req}) as in \cite[Section 3]{ElGamal}, we can obtain\footnote{When $\mb{v}$, $\mb{w}$ are $k\times 1$ vectors, $\mb{u}=\mb{v}\mb{w}$ is also a $k\times 1$ vector such that $u_i = v_iw_i$, $i \in [k]$. }
\begin{align}
 \label{ineq:proof_dist_req}
  &\left(\frac{1}{n} 
\sum_{j = 1}^{n} E\left[(\mb{X}_{j} - \mb{\wh{X}}_{\mb{i}j})(\mb{X}_{j} - \mb{\wh{X}}_{\mb{i}j})^T\right] \right)_d
\notag \\
& \le \delta(\epsilon)\mb{D_i} + (1+\delta(\epsilon))\mb{D_i} + \epsilon'\mb{1} + 2\sqrt{(\delta(\epsilon))\mb{D_i}((1+\delta(\epsilon))\mb{D_i} + \epsilon'\mb{1})}
 \notag \\
 &= \mb{D_i} + 2\delta(\epsilon)\mb{D_i} + \epsilon'\mb{1} + 2\sqrt{(\delta(\epsilon))\mb{D_i}((1+\delta(\epsilon))\mb{D_i} + \epsilon'\mb{1})}\mbox{ for all } i \in [m],
 \end{align}
for sufficiently large $n$.

Thus for all sufficiently large $n$, there exists a code whose rate does not exceed
\begin{align*}
R^G_{LP}(\mb{D},p) + \epsilon'(\epsilon) +\gamma(\epsilon)
\end{align*}
and whose distortion at the decoder $i$ is dominated by the expression in (\ref{ineq:proof_dist_req}). It follows that $R^G_{LP}(\mb{D},p)$ is $\mb{D}$-achievable.
\end{proof}

\section{An LP Lower Bound to $R(\mb{D})$}
\label{sec:lower_bound}
We present our second main result, a lower bound on the rate-distortion function $R(\mb{D})$ of the problem where the source $\mbt{X}$ and side information $\mbt{Y_i}$ are random vectors and the distortion constraint for each decoder $i$ is  $\mb{d_i}(\mbt{X},\mbt{\wh{X}_i}) \le \mbt{D_i}$. The same definitions for the scalar case are used to formulate this problem by replacing the scalar source, side information and distortion constraints by the vector ones given above. 
\begin{definition}
\cite{timo_lessnoisy}
 $\mbt{B}$ is \emph{conditionally less noisy} than $
\mbt{A}$ given $\mbt{C}$, denoted as $(\mbt{B} \succeq \mbt{A} | \mbt{C})$, if
$I(\mbt{W};\mbt{B} | \mbt{C}) \ge I(\mbt{W}; \mbt{A} | \mbt{C})$ for all $\mbt{W}$ such that $\mbt{W} \lra (\mbt{X},\mbt{C}) \lra (\mbt{A},\mbt{B})$.
\end{definition}

\begin{definition}
Given a random vector $\mbt{W}$, $\mathcal{C}(\mbt{W})$ denotes the set of joint distributions over two vectors where the first vector has the same marginal distribution as $\mbt{W}$. 
\end{definition}
We informally refer to $\mathcal{C}(\mbt{W})$ as the ``set of random vectors coupled to $\mbt{W}$" and we sometimes write $\mbt{V} \in \mathcal{C}(\mbt{W})$ to denote such a random vector.

\begin{definition}
Given $\mbt{V} \in \mathcal{C}(\mbt{X})$ and a mapping $\mbt{U}_{\cdot} : $  $\mathcal{C}(\mbt{X}) \rightarrow \mathcal{C}(\mbt{X}, \mbt{V})$, let $R^{LP}_{lb}(\epsilon)$
denote the infinite-dimensional $LP$ in Table \ref{table1}, where $K(\cdot)$ varies over all maps from $\mathcal{C}(\mbt{X})$ to $[0, \infty)$, and $f_1(\cdot)$ and $f_2(\cdot)$  are deterministic functions. Here $K(\cdot)$ assigns the same number to all deterministic random variables and $K(\emptyset)$ denotes this common number. Whenever $(\mbt{U_A},\mbt{V},\mbt{X},\mbt{A},\mbt{B})$ appear together, their joint distribution is assumed to satisfy $(\mbt{U_A},\mbt{V}) \lra \mbt{X} \lra (\mbt{A},\mbt{B})$.
\end{definition} 
\begin{table*}[t]
\caption{LP for the Rate-Distortion Problem}
\label{table1}
$\inf K(\emptyset) - \epsilon \mbox{ \textit{subject to}}$ 
\\
$K(\mbt{X}) = 0  \mbox{  \textit{(initialize)}} $
\\
$K(\mbt{A}) \ge 0, \mbox{ for all } \mbt{A}  \mbox{  \textit{(non-negativity)}}$ 
\\
$K(\mbt{B}) + I(\mbt{B}; \mbt{V}, \mbt{U_B}| \mbt{A}) \ge K(\mbt{A}), \mbox{ for all }(\mbt{A},\mbt{B}) \mbox{ couplings }: \mbt{A} \lra \mbt{B} \lra \mbt{X}  \mbox{ \textit{(slope)}}$
\\
$K((\mbt{A},\mbt{C})) \ge K((\mbt{B},\mbt{C})), \mbox{ for all } (\mbt{A},\mbt{B}) \mbox{ couplings } : (\mbt{B} \succeq \mbt{A} | \mbt{C})  \mbox{ \textit{(monotonicity)}}$
\\
$ K(\mbt{A})  \ge K(\mbt{B}) + I(\mbt{B}; \mbt{V},\mbt{U_A} | \mbt{A}), \mbox{ for all } (\mbt{A},\mbt{B}) \mbox{ couplings } : \mbt{A} \lra \mbt{B} \lra \mbt{X}  \mbox{ \textit{(monotonicity+)}}$
\\
$K(\mbt{A}) + K(\mbt{B}) \ge K(\mbt{C}) + K((\mbt{A}, \mbt{B})), \mbox{ for all }  (\mbt{A},\mbt{B},\mbt{C}) \mbox{ couplings } : \mbt{B}
\lra \mbt{C} \lra \mbt{A} \mbox{ and } \mbt{C} = f_1(\mbt{A}) \mbox{ or } 
 \mbt{C} = f_2(\mbt{B})  \mbox{ \textit{(submodularity)}}$
\end{table*}
\begin{theorem}
\label{theorem:lowergen}
For any $\epsilon > 0$, $R(\mb{D})$ is lower bounded by
\begin{align}
\label{eq:lower_gen}
R_{lb}(\mb{D} + \epsilon\mb{1}) = \inf_{\mbt{V} \in \mathcal{C}(\mbt{X})} \inf_{\mbt{U}_{\cdot}  : \mathcal{C}(\mbt{X}) \rightarrow \mathcal{C}(\mbt{X}, \mbt{V})} 
R^{LP}_{lb}(\epsilon) 
\end{align}
where  $\mbt{V}$ and $\mbt{U}_{\cdot}$ in the infima must satisfy
\\
1) For all $\mbt{B} \in \mathcal{C}(\mbt{X})$, $\mbt{U_B}$ is independent of $\mbt{X}$.
\\
2) If $\{\mbt{B}, A_1, \ldots, A_s \}$ are all elements of $\mathcal{C}(\mbt{X})$ and can be coupled so that $\mbt{X} \lra \mbt{B} \lra (A_1, \ldots, A_s)$ then it must be possible to couple $\mbt{U_B}$ and $(U_{A_1}, \ldots, U_{A_s})$ to $(\mbt{X},\mbt{V})$ such that 
$\mbt{X} \lra (\mbt{V},\mbt{U_B}) \lra (U_{A_1}, \ldots, U_{A_s})$.
\\
3) There exist functions $g_1(\mbt{V},\mbt{U_{Y_1}}, \mbt{Y_1}),\ldots,g_m(\mbt{V},\mbt{U_{Y_m}}, \mbt{Y_m})$ such that 
$E[d_i(\mbt{X}, g_i(\mbt{V},\mbt{U_{Y_i}}, \mbt{Y_i}))] \le \mbt{D_i} + \epsilon\mb{1}$, for all $i \in [m]$. 
\end{theorem}

\begin{proof}
Let $R$ be a $\mb{D}$-achievable rate, $\epsilon > 0$,
 $p(\mbt{x})$ be given and $p(\mbt{y_i} | \mbt{x})$, $i \in [m]$  be fixed.
Then there exists an $(n, M, \mbt{D} + \epsilon \mbt{1})$ code for some $n$ such that $H(I_0) \le n(R + \epsilon)$, where $I_0$ is the output of the encoder. Also, let $K(\mbt{A}) = \frac{I(\mbt{X}^n; I_0| \mbt{A}^n)}{n}$, where $\mbt{A}$ is a random vector with pmf $\sum_{\mbt{x} \in \mathcal{X}}p(\mbt{a}|\mbt{x})p(\mbt{x})$, i.e., $\mbt{A} \in \mathcal{C}(\mbt{X})$.
We call such $\mbt{A}$  \textit{generalized side information}.  Lastly, let $\mbt{V'}_i = I_0 $, $\mbt{U'}_{\mbt{A}i} =(\mbt{A}^{-}_{i}, \mbt{A}^{+}_{i})$, where $\mbt{A}^{-}_{i} = (\mbt{A}_1, \ldots, \mbt{A}_{i-1})$   and $\mbt{A}^{+}_{i} = (\mbt{A}_{i+1}, \ldots, \mbt{A}_n)$ for $i \in [n]$, and let $T$ denote a random variable that is uniformly distributed on $[n]$ such that it is independent of the source $\mbt{X}$, all generalized side information $\mbt{A}$, $\mbt{U'}_{\mbt{A}i} $, and  $\mbt{V'}_i$. Define $\mbt{U_{A}} =(\mbt{U'_A},T)$, $\mbt{V} = (\mbt{V'}, T)$.
Note that we have 
\begin{align*}
R + \epsilon \ge K(\emptyset). 
\end{align*}
Also, we can write
$I(\mbt{X}^n; I_0| \mbt{X}^n) = 0$ and $I(\mbt{X}^n; I_0| \mbt{A}^n) \ge 0, \mbox{ for all } \mbt{A}$, giving the \textit{(initialize)} and  \textit{(non-negativity)} conditions in the LP.
 
 Let $\mbt{A} \lra \mbt{B} \lra \mbt{X}$. For any such $\mbt{A}$ and $\mbt{B}$ we can write $n(K(\mbt{A}) - K(\mbt{B}))$ as
\begin{align*}
I(\mbt{B}^n;I_0|\mbt{A}^n)
& = \sum^n_{i=1} I(\mbt{B}_{i} ; I_0, \mbt{B}^{-}_{i}, \mbt{A}^{-}_{i}, \mbt{A}^{+}_{i} | \mbt{A}_i)
\\
& \le \sum^n_{i=1} I(\mbt{B}_{i} ; I_0, \mbt{B}^{-}_{i}, \mbt{B}^{+}_{i} | \mbt{A}_i)
\\
&= \sum^n_{i=1} I(\mbt{B}_{i} ; \mbt{V'}_i, \mbt{U'}_{\mbt{B}i} | \mbt{A}_i).
 \end{align*}

Since $T$ is independent of $\mbt{X}, \mbt{V'}$, all generalized side information $\mbt{A}$ and all $\mbt{U'_{B}}$,  we can write
 \begin{align*}
n(K(\mbt{A}) - K(\mbt{B})) &\le \sum^{n}_{i=1}  I( \mbt{B}_i; \mbt{V'}_i, \mbt{U'}_{\mbt{B}i}| \mbt{A}_{i},T=i)
 \\
&=n I( \mbt{B}; \mbt{V'}, \mbt{U'}_{\mbt{B}}|\mbt{A},T)
\\
&=n I( \mbt{B}; \mbt{V}, \mbt{U_{B}}| \mbt{A}),
 \end{align*}
 which gives the (\textit{slope}) constraints in the LP. 
 
 Let $(\mbt{B} \succeq \mbt{A} |  \mbt{C})$. Then for each such coupling of $(\mbt{B}, \mbt{A}, \mbt{C})$, $n(K((\mbt{A},\mbt{C})) - K((\mbt{B},\mbt{C})))$ is equal to 
 \begin{align*}
  & H( I_0|\mbt{A}^n,\mbt{C}^n) - H(I_0|\mbt{B}^n,\mbt{C}^n) \ge 0 , \mbox{ by  \cite[Lemma 1]{timo_lessnoisy}},
 \end{align*}
giving the (\textit{monotonicity}) constraints in the LP.

Now we obtain the \textit{monotonicity+} conditions in the LP. Let $\mbt{A}\lra \mbt{B} \lra \mbt{X}$. By utilizing the chain rule again, we can write 
$n(K(\mbt{A}) - K(\mbt{B}))$ as
\begin{align}
I(\mbt{B}^n;I_0|\mbt{A}^n)
&\ge \sum^n_{i=1} I(\mbt{B}_{i} ; I_0, \mbt{A}^{-}_{i}, \mbt{A}^{+}_{i} | \mbt{A}_i) 
\notag \\
\label{eq:decodeA}
&=\sum^n_{i=1} I(\mbt{B}_{i} ; \mbt{V'}_i, \mbt{U'}_{\mbt{A}i} | \mbt{A}_i) 
\notag \\
& = n I(\mbt{B} ; \mbt{V}, \mbt{U_A} | \mbt{A}), 
\end{align}
giving  (\textit{monotonicity+}) conditions.

 Let $\mbt{A},\mbt{B},\mbt{C}$ be such that
$\mbt{A} \lra \mbt{C} \lra \mbt{B}$ and $\mbt{C} = f_1(\mbt{A})$ for some deterministic mapping $f_1(\cdot)$. By the chain rule, $n(K(\mbt{A}) + K(\mbt{B}))$ is equal to
\begin{align*}
&I(\mbt{B}^n; I_0|\mbt{A}^n) + I(\mbt{X}^n; I_0|\mbt{B}^n, \mbt{A}^n) + I(\mbt{X}^n; I_0|\mbt{B}^n)
\\
&\ge I(\mbt{B}^n; I_0|\mbt{C}^n) + I(\mbt{X}^n; I_0|\mbt{B}^n, \mbt{A}^n) + I(\mbt{X}^n; I_0|\mbt{B}^n)
\\
&\ge I(\mbt{B}^n; I_0|\mbt{C}^n) + I(\mbt{X}^n; I_0|\mbt{B}^n, \mbt{A}^n) + I(\mbt{X}^n; I_0|\mbt{B}^n,\mbt{C}^n)
\\
&=   I(\mbt{X}^n,\mbt{B}^n; I_0|\mbt{C}^n) + I(\mbt{X}^n; I_0|\mbt{B}^n, \mbt{A}^n) 
\\
&=   I(\mbt{X}^n; I_0|\mbt{C}^n) + I(\mbt{X}^n; I_0|\mbt{B}^n, \mbt{A}^n).
\end{align*} 
By setting $\mbt{C}= f_2(\mbt{B})$ and swapping the role of $\mbt{A}$ and $\mbt{B}$ in the procedure above,  we get the (\textit{submodularity}) conditions.

Now we find the properties of  $\mbt{V}$ and $\mbt{U_{A}}$ that give us the conditions 1)--3) in Theorem \ref{theorem:lowergen} and  the Markov chain property in Definition $5$.
Let $\mbt{A}$, $\mbt{\bar{A}} = (A_1,\ldots, A_s)$ for some $s$ be such that $\mbt{X} \lra \mbt{A} \lra (A_1,\ldots, A_s)$. 
Firstly, since any set of $\mbt{U'}_{\mbt{A}i}$ is independent of $\mbt{X}_i$ and of any set of generalized side information $A_i$'s, all $\mbt{U_{A}}$'s are independent of  $\mbt{X}$ and all $\mbt{A}$'s. Secondly, note that $\mbt{X}_i \lra ( \mbt{V'}_i, \mbt{U'}_{\mbt{A}i}) \lra (\mbt{V'}_i, {U'}_{A_1i}, \ldots {U'}_{A_si}) $ since
\begin{align*}
H(\mbt{V'}_i, \mbt{U'}_{\bar{\mbt{A}}i}|\mbt{V'}_i, \mbt{U'}_{\mbt{A}i},\mbt{X}_i) &= H(\mbt{\bar{A}}^{-}_{i}, \mbt{\bar{A}}^{+}_{i}| I_0, \mbt{A}^{-}_{i}, \mbt{A}^{+}_{i}, \mbt{X}_i)
\\
&= H(\mbt{\bar{A}}^{-}_{i}, \mbt{\bar{A}}^{+}_{i}| I_0, \mbt{A}^{-}_{i}, \mbt{A}^{+}_{i}).
\end{align*}

Then $\mbt{X}\lra (\mbt{V'}, \mbt{U'_{A}}) \lra (\mbt{V'}, {U'}_{A_1}, \ldots {U'}_{A_s})$ implies  $\mbt{X} \lra (\mbt{V}, \mbt{U_{A}})\lra (\mbt{V}, U_{A_1}, \ldots U_{A_s})$.
Furthermore, given  $(\mbt{V}, \mbt{U_{Y_i}})$ and $\mbt{Y_i}$, $i \in [m]$, decoder $i$ can reconstruct the source subject to its own distortion constraint. Lastly, $(\mbt{V}, \mbt{U_{A}}) \lra \mbt{X} \lra (\mbt{A}, \mbt{B})$ since $(\mbt{V'}_i, \mbt{U'}_{\mbt{A}i}) \lra \mbt{X}_i \lra (\mbt{A}_i, \mbt{B}_i) $ for all $i \in [n]$. 
\end{proof}

We can interpret $K(\mbt{A})$ in the LP as the amount of information that
a hypothetical decoder with side information $\mbt{A}$ receives about $\mbt{X}$
from the broadcasted message. We can also view $\mbt{U_A}$ as a quantized representation of the source that the hypothetical decoder can extract from the message with the help of its side information $\mbt{A}$ and $\mbt{V}$ as a common message to all decoders.

The \textit{(submodularity)} condition is so named for the following reason. Let $\mbt{X} = (X_1, \ldots, X_k)$, where $X_i$'s are all independent random variables  and let $\mbt{A} \subseteq \mbt{X}$,  $\mbt{B} \subseteq \mbt{X}$ \footnote{Although $\mb{X}$ is a vector, we can view it as an ordered set which also induces an ordered set structure on the subsets. Hence, we can use the set notation whenever it is convenient.}. Then we can write the \textit{(submodularity)} condition for such $\mbt{A}$ and $\mbt{B}$ as
$K(\mbt{A}) + K(\mbt{B}) \ge K(\mbt{A}\cap \mbt{B}) + K(\mbt{A} \cup \mbt{B})$.
The LP lower bound was inspired by a similar lower bound for the special case of
index coding~\cite{blasiak}. That lower bound does not require auxiliary
random variables, and it is expressed in terms of entropy instead of
mutual information.

\begin{remark}
Evidently other conditions that $I(\mbt{X}^n;I_0|\mbt{A}^n)/n$ must
satisfy in the context of this problem can be 
incorporated into the bound as desired.
\end{remark}

\begin{remark}
The lower bound in Theorem \ref{theorem:lowergen} can be generalized to continuous sources with well-behaved distortion constraints such as Gaussian sources subject to component-wise mean square error (MSE) distortion constraints.
\end{remark}


The lower bound in Theorem~\ref{theorem:lowergen} is not evidently
computable, since the infimum over $K(\cdot)$ is subject to a 
continuum of constraints and there are no cardinality bounds
on the $V$ and $U_\cdot$ variables. We next provide a weakened
lower bound that is computable. For this we need the following
notation.

\begin{notation}
\label{notation:computable}
Let  $\mbt{A} \lra \mbt{B} \lra \mbt{X}$ and $\mbt{D_A} = \{\mbt{D_i} | \mbt{Y_i} \lra \mbt{A} \lra \mbt{X} \}$. Then $R(\mbt{D_A})$ denotes the result of the following optimization problem :
\begin{align*}
\min_{C_\mbt{A}} I(\mbt{B}; \mbt{V} | \mbt{A} )
\end{align*}
where 
\begin{align*}
C_{\mbt{A}} : &\mbt{V} \in \mathcal{C}(\mbt{X}) \mbox{ such that }
\\
&\mbox{ there exists functions } g_i(\mbt{V}, \mbt{Y_i}) \mbox{ such that }
E[d_i(\mbt{X}, g_i(\mbt{V}, \mbt{Y_i}))] \le \mbt{D_i} \mbox{ for all }  \mbt{D_i} \in \mbt{D_A}.
\end{align*}
\end{notation}


\begin{theorem}
\label{theorem:lowergen_comp}
Let $S_{\mbt{A}}$ be a finite set of generalized side information variables $\mbt{A} \in C(\mbt{X})$ and consider the function $K(\cdot)$ over the elements of $S_{\mbt{A}}$. For any $\epsilon > 0$, $R_{lb}(\mb{D} + \epsilon \mb{1})$ is lower bounded by $R'_{lb}(\mb{D} + \epsilon \mb{1})$ where $R'_{lb}(\mb{D} + \epsilon \mb{1})$ is equal to
\begin{align}
\label{lin_prog:relaxed}
 \inf K(\emptyset) - \epsilon,
\end{align}
where the infimum is over all $K(\cdot ) : S_{\mbt{A}} \rightarrow [0, \infty)$ such that
\\
$K(\mbt{X}) = 0  \mbox{  \textit{(initialize)}} $
\\
$K(\mbt{A}) \ge 0, \mbox{ for all } \mbt{A}  \mbox{  \textit{(non-negativity)}}$ 
\\
$K((\mbt{A},\mbt{C})) \ge K((\mbt{B},\mbt{C})), \mbox{ for all } (\mbt{A},\mbt{B}) : (\mbt{B} \succeq \mbt{A} | \mbt{C})  \mbox{ \textit{(monotonicity)}}$
\\
$ K(\mbt{A})  \ge K(\mbt{B}) + R(\mbt{D_A} + \epsilon\mbt{1}), \mbox{ for all } (\mbt{A},\mbt{B})  : \mbt{A} \lra \mbt{B}  \lra \mbt{X} \mbox{ \textit{(monotonicity+)}}$
\\
$K(\mbt{A}) + K(\mbt{B}) \ge K(\mbt{C}) + K((\mbt{A}, \mbt{B})), \mbox{ for all }  (\mbt{A},\mbt{B},\mbt{C}) \mbox{ couplings } : \mbt{B}
\lra \mbt{C} \lra \mbt{A} \mbox{ and } \mbt{C} = f_1(\mbt{A}) \mbox{ or } 
 \mbt{C} = f_2(\mbt{B})  \mbox{ \textit{(submodularity)}}$
\end{theorem}

\begin{proof}[Proof of Theorem \ref{theorem:lowergen_comp}]
Let $\epsilon >0$, $\mbt{V} \in \mathcal{C}(\mbt{X})$, and $\mbt{U}. : \mathcal{C}(\mbt{X}) \rightarrow \mathcal{C}(\mbt{X},\mbt{V})$ satisfying the conditions 1)--3) in Theorem \ref{theorem:lowergen} be given. Also, let LP1 be the linear program in Table \ref{table1} when $\mbt{A}$, $\mbt{B}$ and $\mbt{C}$ are in $S_{\mbt{A}}$ and let the solution of LP1 be denoted by $\bar{R}^{LP}_{lb}(\epsilon)$.
Then $R^{LP}_{lb}(\epsilon)$ in Theorem \ref{theorem:lowergen} is lower bounded by
$\bar{R}^{LP}_{lb}(\epsilon)$. Therefore it is enough to show $\bar{R}^{LP}_{lb}(\epsilon) \ge R'_{lb}(\mb{D} + \epsilon \mb{1})$.  Note that the constraints in LP1 and the LP in Theorem \ref{theorem:lowergen_comp}, denoted by LP2,  are the same except the \textit{monotonicity+} condition is different and there is no $slope$ condition in LP2. But for any $\mbt{A} \lra \mbt{B} \lra \mbt{X}$
the \textit{monotonicity+} condition in LP1 implies the \textit{monotonicity+} condition in LP2 since $I(\mbt{B}; \mbt{V}, \mbt{U_A} | \mbt{A}) \ge 
R(\mbt{D_A} + \epsilon\mbt{1})$ by condition 2) and 3) in Theorem \ref{theorem:lowergen}. Hence, $\bar{R}^{LP}_{lb}(\epsilon) \ge R'_{lb}(\mb{D} + \epsilon \mb{1})$.
\end{proof}
Note that $R'_{lb}(\mb{D} + \epsilon\mb{1})$ is computable since we have a finite number of constraints in the LP and each $R(\mb{D_A})$ can be computed by finding a cardinality constraint on the auxiliary random variable $\mbt{V}$ using standard techniques \cite{Csiszar}. 

\section{Comparison with Other Bounds}
\label{sec:subsumes}
\subsection{Upper Bound}

Although there are achievable schemes for various forms of
rate-distortion with side information (e.g. \cite{watanabe},\cite{sinem_index},\cite{timo},\cite{berger}), most are for special cases of the problem. The two exceptions, both of which purport to provide achievable schemes for the general problem considered here, are Heegard and Berger~\cite{berger} and Timo \emph{et al.}~\cite{timo}. Heegard and Berger's achievable rate was shown to be incorrect
via a counterexample by Timo \emph{et al.}, who also to provided a corrected
scheme. In fact, the proof of Timo \emph{et al.}'s achievable result 
contains an error that is similar to that of Heegard and Berger.
To see this, let us state Timo \emph{et al.}'s achievable 
result.\footnote{This problem also afflicts Theorem 1 in Timo \emph{et al.},
although we shall focus our discussion on Theorem 2 of that paper,
which is simpler and directly comparable to Theorem~\ref{thm:gen_ach} 
in the present paper.}
\begin{notation}
$\bar{v} = \set{S}_1, \ldots, \set{S}_{2^m-1}$ denotes an ordered list of all possible nonempty subsets of $[m]$, where each $\set{S}_i$ denotes a different subset such that $|\set{S}_i| \ge |\set{S}_j|$ for all $i  < j$.  $\bar{\set{V}}$ denotes the set of all possible such $\bar{v}$. 
\end{notation}

\begin{notation}
\begin{align*}
U^{-'}_{\set{S}_j} &= \Big\{U_{\set{S}_i} \in \set{U}\ |\ i < j,\ \set{S}_i \nsupseteq \set{S}_j  \Big\},
\\
U^{\supset}_{\set{S}_j} &=\Big\{ U_{\set{S}_i} \in \set{U}\ |\ \set{S}_i \supset \set{S}_j \Big\},
\\
U^{+}_{\set{S}_j} &= \Big\{U_{\set{S}_k} \in \set{U}\ |\ k > j,\ \set{S}_k \cap \set{S}_j \neq \emptyset \Big\}, 
\\
U^{\dag}_{\set{S}_j} &= \left\{ U_{\set{S}_i} \in U^{-'}_{\set{S}_j} |
  \begin{array}{ll}
    \exists U_{\set{S}_k} \in U^+_{\set{S}_j}, \\
    \set{S}_i \cap \set{S}_k \neq \emptyset
  \end{array}\right\}, \mbox{ and } 
  \\
  U^\ddag_{\set{S}_j,l} &= \Big\{U_{\set{S}_i} \in U^\dag_{\set{S}_j}\ :\ \set{S}_i\ni l \Big\}\ \text{ when } l \in \set{S}_j.
\end{align*}

\end{notation}
 
\begin{claim}[Theorem 2,\cite{timo}]
\label{thm:ach_timo}
The rate-distortion function $R(\mb{D})$ is upper bounded by 
\begin{align}
\label{eq:ach_timo}
&R^T_{ach}(\mb{D}) = \min_{\bar{v} \in \bar{\set{V}}} \inf_{C_{ach,v}(\mb{D})}
\inf_{C^{LP}_{T}}  \sum^{2^m-1}_{j = 1} R_{\set{S}_j} ,
\end{align}
where $C_{ach,\bar{v}}(\mb{D})$ is as in Theorem \ref{thm:gen_ach}
and
\begin{align}
 C^{LP}_{T} :\  & 1) R_{\set{S}_j} \ge 0,   R'_{\set{S}_j} \ge 0  \mbox{ for all } j \in [2^m-1] 
\notag \\
& 2) R_{\set{S}_j} \ge I\big(X,U^\dag_{\set{S}_j},U^\supset_{\set{S}_j};U_{\set{S}_j}\big) - R'_{\set{S}_j}
\mbox{ for all } j \in [2^m-1]
\notag \\
& 3)  R'_{\set{S}_j} \le\min_{l \in \set{S}_j}
I\big(U_{\set{S}_j};U^\ddag_{\set{S}_j,l},U^\supset_{\set{S}_j},Y_{l}\big) \mbox{ for all  } j \in [2^m -1].
\notag 
\end{align}
\end{claim}

The proof given by 
Timo \emph{et al.} proceeds as follows.
Let $\bar{v} \in \bar{\set{V}}$ be given. The codebook generation is the same as in the proof of Theorem \ref{thm:gen_ach}. Encoding is almost the same except that  at each stage $j$, we select a codeword that is jointly typical with only those already-selected codewords that correspond to 
the messages $U^\dag_{\set{S}_j},U^\supset_{\set{S}_j}$ and the source, instead of messages $U^-_{\set{S}_j}$ and the source as in Theorem~\ref{thm:gen_ach}.
This creates an issue, however, because if the encoding proceeds
in this fashion then there is no guarantee that the
variables $U^\dag_{\set{S}_j},U^\supset_{\set{S}_j}$ are themselves
jointly typical.

To illustrate this, consider the case in which there are six decoders and suppose that $\bar{v}=[6], \ldots,\{1,2\}, \{5,6\},\{3,4\},\{2,3\},\{4,5\},\{6\},\{5\}$, $\{4\},\{3\},\{2\},\{1\}$. Choose $\set{U}$ such that all $U_{\set{S}_j} = \emptyset$ except $U_{\{i,i+1\}}$, for $i \in [5]$. Then the encoding order of the nontrivial messages is $(U_{\{1,2\}}, U_{\{5,6\}},U_{\{3,4\}},U_{\{2,3\}},U_{\{4,5\}})$. When the message $U_{\{3,4\}}$ is encoded, the encoder selects a codeword that is jointly typical with the codewords related to messages $ U^\dag_{\{3,4\}}= (U_{\{1,2\}},U_{\{5,6\}})$ and the source (note that $U^\supset_{\{3,4\}} = \emptyset$).
However, in previous stages $U_{\{1,2\}}$ and $U_{\{5,6\}}$ were not selected in a way that guarantees that they are jointly typical, since  $U_{\{1,2\}} \notin \{ U^\dag_{\{5,6\}} \cup U^\supset_{\{5,6\}} \}$ and $U_{\{5,6\}} \notin \{U^\dag_{\{1,2\}} \cup U^\supset_{\{1,2\}} \}$.
The rate analysis in Timo \emph{et al.}, specifically the use of Lemma~3 
in that 
paper, presumes that the codewords corresponding to $U_{\{1,2\}}$ and 
$U_{\{5,6\}}$ are jointly typical when the codeword for $U_{\{3,4\}}$
is chosen. This error is similar to the one in Heegard and 
Berger~\cite{berger}.\footnote{Unlike the Heegard-Berger
result, however, the rate promised by Timo \emph{et al.}'s achievable 
result is not known to be unachievable in general at this point.}
For the two-decoder case, this issue does not arise, and the
Timo \emph{et al.} rate is indeed achievable, as is that of Heegard and 
Berger.

This error could be fixed in several ways. Our scheme in 
Theorem \ref{thm:gen_ach} avoids this issue by requiring that
each codeword be jointly typical with all of the previously-selected
codewords. If a certain pair of auxiliary random variables never 
appear together in any of the mutual information expressions, then
one can impose a conditional independence condition between them
without loss of generality, which is tantamount, from a rate perspective,
to not requiring that they be chosen in a way that ensures their
joint typicality.

 Our scheme in Theorem \ref{thm:gen_ach} differs from the achievable scheme in \cite{timo} in two other respects as well. We do not require that the
sets in $v$ be ordered so that their cardinalities are nonincreasing.
Arguably the most notable difference is in the decoding. While 
in \cite{timo}, each decoder decodes its messages sequentially in the same 
order that they are encoded, in our scheme we apply simultaneous decoding, i.e., we decode all messages for decoder $i$ together. We shall see later,
when discussing the odd-cycle index coding problem in 
Section~\ref{sec:odd_cycle_gauss}, that for a given
class of auxiliary random variables, simultaneous decoding can yield
a strict rate improvement.

We conclude this subsection by showing that for the two-decoder case in which Claim \ref{thm:ach_timo} is valid, the upper bound in \cite{timo} is 
no worse than that of Theorem~\ref{thm:gen_ach}.
 
\begin{lemma}
\label{lemma:ach_2dec}
When there are two decoders, $R_{ach}(\mb{D})$ is upper bounded by 
\begin{align}
R^T(\mb{D})  = \min_{C_{ach,v}(\mb{D})} \max_{i \in  \{1,2\}}\{I(X; U_{\{1,2\}} | Y_i)\} + I(X; U_{\{1\}} | U_{\{1,2\}}, Y_1) +  I(X; U_{\{2\}} | U_{\{1,2\}}, Y_2),
\label{eq:Timo:two}
\end{align}  
where $C_{ach,v}(\mb{D})$ is in Theorem \ref{thm:gen_ach}.
\end{lemma}

\begin{proof}[Proof of Lemma \ref{lemma:ach_2dec}]
Firstly notice that $U_{\{1\}}$ and $U_{\{2\}}$ never appear together 
on the right-hand side of~(\ref{eq:Timo:two}). Hence without loss of 
optimality we can add the condition 
$U_{\{1\}} \perp U_{\{2\}} | X, U_{\{1,2\}}$ to $C_{ach,v}(\mb{D})$.
Let  $v = \{\{1,2\}, \{1\}, \{2\} \}$ and $U_{\set{S}_j} \in 
C_{ach,v}(\mb{D})$ with $U_{\{1\}} \perp U_{\{2\}} | X, U_{\{1,2\}}$.
From the LP conditions in Theorem \ref{thm:gen_ach}, we can write
\begin{align}
&R_{\{1,2\}} + R'_{\{1,2\}} \ge I(X; U_{\{1,2\}})
\\
& R_{\{1\}} + R'_{\{1\}} \ge I(X, U_{\{1,2\}};  U_{\{1\}})
\\
& R_{\{2\}} + R'_{\{2\}} \ge I(X, U_{\{1,2\}}, U_{\{1\}};  U_{\{2\}})
\\
& R'_{\{i\}} \le I( U_{\{i\}}; U_{\{1,2\}},  Y_i), \mbox{ for all } i \in \{1,2\}
\\
& R'_{\{1,2\}} \le \min_{ i \in \{1,2\}} \{ I( U_{\{1,2\}}; U_{\{i\}}, Y_i ) \}
\\
& R'_{\{1,2\}} + R'_{\{i\}}  \le I( U_{\{1,2\}};Y_i) + I(U_{\{i\}}  ; U_{\{1,2\}},Y_i), \mbox{ for all } i \in \{1,2\}.
\end{align}
Then  $R'_{\{1,2\}} = \min_{ i \in \{1,2\}} \{ I( U_{\{1,2\}}; Y_i ) \}$,
$R'_{\{i\}} =  I( U_{\{i\}}; U_{\{1,2\}},  Y_i)$, 
$R_{\{1,2\}} + R'_{\{1,2\}} = I(X; U_{\{1,2\}})$,
$ R_{\{1\}} + R'_{\{1\}} = I(X, U_{\{1,2\}};  U_{\{1\}})$, and
$ R_{\{2\}} + R'_{\{2\}} = I(X, U_{\{1,2\}}, U_{\{1\}};  U_{\{2\}})$ are feasible choices enabling us to  upper bound  $\inf_{C^{LP}_{ach}} \sum^3_{j =1} R_{\set{S}_j}$ in (\ref{eq:gen_ach}) by
\begin{align}
\max_{i \in  \{1,2\}}\{I(X; U_{\{1,2\}} | Y_i)\} + I(X; U_{\{1\}} | U_{\{1,2\}}, Y_1) +  I(X; U_{\{2\}} | U_{\{1,2\}}, Y_2) + I(U_{\{1\}} ; U_{\{2\}}| X, U_{\{1,2\}}),
\end{align}
which is equal to the mutual information expression in Lemma \ref{lemma:ach_2dec} when 
$U_{\{1\}} \perp U_{\{2\}} | X, U_{\{1,2\}}$. Therefore, $R^T(\mb{D}) \ge R_{ach}(\mb{D})$.
\end{proof}

\subsection{Lower Bounds}
\subsubsection{\textit{minimax}-type Lower Bound}
First we compare the general lower bound, $R_{lb}(\mb{D}+\epsilon\mb{1})$, with the minimax version of the  lower bound in \cite{sinem_index}.  For completeness, we state the minimax version of the theorem below. 
\begin{theorem}
\label{theorem:minmax}
Let the pmf's $p(x, y_i)$ for all $i \in [m]$ be given. Then $R(\mb{D})$ is lower bounded by
 \begin{align}
 \label{3lower_general}
 &R^m_{lb}(\mb{D} + \epsilon \mb{1}) = \sup_{\bar{P}} \inf_{\bar{C}} \bar{R}_{lb} - \epsilon,
\\
\label{4lower_general}
\mbox{where }&\bar{R}_{lb} =\max_{\sigma} \big{[} I(X; V, U_{Y_\sigma(1)}|Y_{\sigma(1)})
 \notag \\
& \quad + I( X ; U_{Y_\sigma(2)}| V, U_{Y_\sigma(1)},Y_{\sigma(1)},Y_{\sigma(2)}) +\cdots
\notag \\
&\quad + I( X ; U_{Y_\sigma(m)}| V, U_{Y_\sigma(1)},\ldots, U_{Y_\sigma(m-1)},{Y}) \big{]},
\end{align}
${Y} = (Y_{\sigma(1)}, \ldots,Y_{\sigma(m)})$, and
 \\
1)  $\bar{P}=\{ p(x,y_1, \ldots, y_m) | \sum_{\\
y_j : j\neq i  }p(x,y_1, \ldots, y_m)$ $= p(x,y_i), \forall i \in [m]\}$.
\\
2) $\bar{C}$ denotes the set of $(V, U_{Y_1},\ldots,U_{Y_m})$ jointly distributed with $X, Y_1,\ldots,Y_m$ such that
\\
$ (Y_1,\ldots, Y_m) \lra X \lra ( V, U_{Y_1}, \ldots, U_{Y_m})$
and there exists functions $g_1, \ldots, g_m$ with the property that
 \\
$\mathbb{E}[d_i(X,g_{i}( V, U_{Y_i}, Y_i))] \le D_i
+ \epsilon, \forall i \in [m]$.
\\
3) $\sigma(.)$ denotes a permutation on integers $[m]$.
\end{theorem}
The minimax lower bound in Theorem \ref{theorem:minmax}  is the state-of-the-art for the general rate-distortion problem with side information at multiple decoders.
Note that in Theorem \ref{theorem:minmax}, one can absorb $V$ into $U_{Y_i}$, $i \in [m]$ without loss of optimality. For the ease of comparison with $R_{lb} (\mb{D} + \epsilon \mb{1})$ we leave it as a separate variable, however.
\begin{theorem}
\label{theorem:subs_minmax}
$R_{lb} (\mb{D} + \epsilon \mb{1})\ge R^m_{lb}(\mb{D} + \epsilon \mb{1})$, where $\epsilon > 0$.
\end{theorem}

\begin{proof}
Consider  $R_{lb}(\mb{D} + \epsilon \mb{1})$. Note that the $LP$ constraints of $R_{lb}(\mb{D} + \epsilon \mb{1})$ apply to all choices of the relevant
random variables. Hence
we can write
\begin{align}
\label{ineq:relaxedlb1}
R_{lb}(\mb{D} + \epsilon \mb{1}) \ge  \sup_{\bar{P}} \inf_{V \in \mathcal{C}(X)} \inf_{U_{\cdot}  : \mathcal{C}(X) \rightarrow \mathcal{C}(X, V)}  {R}^{LP}_{lb}(\epsilon)
\end{align}
where $\bar{P}$ is as in Theorem \ref{theorem:minmax}, and $V$ and $U$ in the infima satisfy the conditions 1)--3) in Theorem \ref{theorem:lowergen} for a fixed coupling of the random variables. Now we find a lower bound to the 
quantity
$R^{LP}_{lb}(\epsilon)$ in (\ref{ineq:relaxedlb1}) by utilizing the \textit{monotonicity} and \textit{monotonicity+} constraints of the LP in Table \ref{table1}. We can write the following series of inequalities:
\begin{align}
\label{ineq:LPcon0}
K(\emptyset) & \ge K(Y_1) \mbox{ by \textit{(monotonicity)}}
\\
\label{ineq:LPcon1}
K(Y_1) & \ge  K(Y_1,Y_2) + I(Y_2; V, U_{Y_1}| Y_1)  
\\
\vdots 
\notag \\
K(Y_1,\dots, Y_m) & \ge K(Y_1,\ldots,Y_m, X) 
\label{ineq:LPcon2}
 + 
I(X;V, U_{Y_1},\ldots,U_{Y_m}| Y_1,\ldots,Y_m)
 \\
\label{ineq:LPcon3}
K(Y_1,\ldots,Y_m, X) & = 0. 
\end{align}
where 
(\ref{ineq:LPcon1}) is from \textit{monotonicity+}  and  (\ref{ineq:LPcon2}) is from \textit{monotonicity+} and 
 $(Y_1,\ldots, Y_m) \lra X \lra (V, U_{Y_1 \ldots Y_m}) \lra (V, U_{Y_1}, \ldots, U_{Y_m})$. 
If we add all these inequalities side-by-side we obtain
\begin{align}
 K(\emptyset)& \ge I(Y_2; V, U_{Y_1}| Y_1)+ \cdots 
 \notag 
 \\
 &\quad + I(Y_m; V, U_{Y_1},\ldots,U_{Y_{m-1}}| Y_1,\ldots,Y_{m-1}) 
 \notag
 \\
 \label{ineq:mm_lbcomp}
 &\quad + I(X; V, U_{Y_1},\ldots,U_{Y_m}| Y_1,\ldots,Y_m).
 \end{align}
 By applying a series of chain rules and combining terms, we can write the right-hand side of (\ref{ineq:mm_lbcomp}) as
 \begin{align*}
 &I(X; V, U_{Y_1}| Y_1) + \cdots + I(X; U_{Y_2}| V, U_{Y_1},Y_1, Y_2) 
 \\
 & \quad + I(X; U_{Y_m} | V, U_{Y_1},\ldots,U_{Y_{m-1}},Y_1,\ldots,Y_m).
 \end{align*}
Let us define 
 \begin{align*}
 \Gamma_k &= \sum^k_{i =2} I(Y_i; V, U_{Y_1},\ldots, U_{Y_{i-1}}| Y_1, \ldots, Y_{i-1})
 \\
 & \quad + I(X; V, U_{Y_1},\ldots,U_{Y_k}| Y_1,\ldots,Y_{k}) 
 \\
 & \quad +\sum^m_{i = k+1}  I(X; U_{Y_i} | V, U_{Y_1},\ldots,U_{Y_{i-1}},Y_1,\ldots,Y_i)
 \end{align*}
 for $k \in [m]$ where ``empty" sums are zero. Note that $\Gamma_m$ is equal to the right-hand side of (\ref{ineq:mm_lbcomp}). One can show that $\Gamma_1 = \Gamma_2$ $=\ldots = \Gamma_m$. Hence
 $K(\emptyset) \ge \Gamma_1$.

  Also since there are $m$ decoders,  we can get $m!$ lower bounds on $K(\emptyset)$ by considering all possible permutations on integers $[m]$.  Hence, we have
$K(\emptyset) \ge \bar{R}_{lb}$. From (\ref{ineq:relaxedlb1}) we can write
\begin{align}
\label{ineq:relaxed_minmax-1}
{R}_{lb}(\mb{D}+ \epsilon\mb{1}) &\ge \sup_{\bar{P}} \inf_{V \in \mathcal{C}(X)} \inf_{U_{\cdot}  : \mathcal{C}(X) \rightarrow \mathcal{C}(X, V)}  \bar{R}_{lb} - \epsilon
 \\
\label{ineq:relaxed_minmax}
 &{\ge} \sup_{\bar{P}} \inf_{\bar{C}} \bar{R}_{lb} - \epsilon,
\end{align}
 where $\bar{C}$ is as in Theorem \ref{theorem:minmax}. Lastly, we have (\ref{ineq:relaxed_minmax}) since each feasible set of random variables in the infima in (\ref{ineq:relaxed_minmax-1}) is also feasible for $\bar{C}$. 
  Hence, 
$R_{lb}(\mb{D} + \epsilon\mb{1}) \ge R^m_{lb}(\mb{D} + \epsilon\mb{1}) $.
\end{proof}

\subsubsection{LP Lower Bound for the Index Coding Problem}
We next compare the general lower bound, $R_{lb}(\mb{D} + \epsilon \mb{1})$ with the linear programming lower bound in \cite{blasiak} for the index coding problem \cite{birk}. In the index coding problem, the source  $\mbt{X} = (X_1,\ldots,X_k)$ is such that $X_i$, $i \in [k]$ are independent and identically distributed (i.i.d.) Bernoulli $\left(\frac{1}{2}\right)$ random variables and each side information $\mbt{Y_i}$ at decoder $i$ is an arbitrary subset\footnote{Although $\mbt{X}$ is a vector, we can view it as an ordered set which also induces an ordered set structure on the subsets. Hence, we can use the set notation whenever it is convenient.} of the source $\mbt{X}$. Each decoder $i$ wishes to reconstruct an arbitrary subset of the source, $ \mbt{{\wh{X}_i}} \subseteq \mbt{X} \setminus \mbt{Y_i}$. The reconstructions can either be required to be zero error \cite{blasiak} or such that the block error probability vanishes~\cite{sinem_index}.
Both formulations are more stringent than considering the problem with Hamming distortion in the limit in which the distortion goes to zero,
so $R_{lb}(\epsilon\mb{1})$ is a valid lower bound to the index coding problem 
in all three cases.

We first state the LP lower bound in \cite{blasiak}, originally stated
for the zero-error form of the problem. For completeness, we need the following notation.
\begin{notation}
$\mbt{A} \leadsto \mbt{B}$ denotes ``$ \mbt{A}$ \textit{decodes} $ \mbt{B}$,'' meaning that $ \mbt{A} \subseteq  \mbt{B}$ and for every source component $X_i \in   \mbt{B} \setminus  \mbt{A}$ there is a decoder $j$  who reconstructs $X_i$  and $ \mbt{Y_j} \subseteq  \mbt{A}$. Also $S(\mbt{A}) = \{ X_i  | \mbox{ decoder $j$ reconstructs $X_i \in \mbt{X}$ and } \mbt{Y_j} \subseteq \mbt{A}  \}$.
\end{notation}

\begin{theorem}[LP lower bound \cite{blasiak}]
\label{theorem:blasiak}

The optimal value for the linear program in Table \ref{table2} \footnote{The statement of the result in \cite{blasiak} does not contain the (monotonicity) condition, although it is clear from the proof that it was intended to be 
included. The condition is present in the
preprint version of the paper \cite{blasiak_arxiv}.}, denoted by $\wh{R}^{LP}_{lb}$, is a lower bound to the index coding problem.

\begin{table*}[t]
\caption{LP Bound for Index Coding Problem }
\label{table2}
$\min \wh{K}(\emptyset) \mbox{ \textit{subject to}} $
\\
$\wh{K}( \mbt{X}) \ge | \mbt{X}|  \mbox{   \textit{(initialize)}}$
 \\
$\wh{K}( \mbt{A}) + | \mbt{B}\setminus  \mbt{A}| \ge \wh{K}( \mbt{B}), \mbox{ for all }  \mbt{A}\subseteq  \mbt{B}\subseteq  \mbt{X}  \mbox{ \textit{(slope)}}$
\\
$\wh{K}( \mbt{B}) \ge \wh{K}( \mbt{A}), \mbox{ for all }  \mbt{A} \subseteq  \mbt{B} \subseteq   \mbt{X}  \mbox{ \textit{(monotonicity)}}$
 \\
$\wh{K}( \mbt{A}) = \wh{K}( \mbt{B}), \mbox{ for all }   \mbt{A}, \mbt{B} \subseteq  \mbt{X}:  \mbt{A} \leadsto  \mbt{B} \mbox{ \textit{(decode)}}$
 \\
$\wh{K}( \mbt{A}) + \wh{K}( \mbt{B}) \ge \wh{K}( \mbt{A}\cap  \mbt{B}) + \wh{K}( \mbt{A} \cup  \mbt{B}),$
\\
$ \mbox{ for all }  \mbt{A}, \mbt{B} \subseteq  \mbt{X} \mbox{ \textit{(submodularity)}}.$
\end{table*}
\end{theorem}
Now that we stated the LP lower bound in \cite{blasiak}, we show that  $\lim_{\epsilon \rightarrow 0}R_{lb}(\epsilon\mb{1})$ is equal to this bound when we restrict the generalized side information, $\mbt{A}
$, in $R_{lb}(\epsilon\mb{1})$ to be a subset of the source, $\mbt{X}$. 
From now on we denote this weakened form of $R_{lb}(\epsilon\mb{1})$ obtained by restricting the generalized side information to be a subset of the source  by $R^I_{lb}(\epsilon\mb{1})$.
The following two lemmas will be useful to prove that the weakened lower bound  $R^I_{lb}(\epsilon\mb{1})$ is equal to the LP lower bound in Theorem \ref{theorem:blasiak}.
\begin{lemma}
\label{lemma:LP_insert}
Without loss of optimality we can  replace the  $\mbox{ (initialize)}$  and $\mbox{ \textit{(slope)}}$ conditions in the LP in Table \ref{table2} with 
\begin{align*}
&\wh{K}(\mbt{X}) = |\mbt{X}| \mbox{ (initialize*)}
\\
& \wh{K}( \mbt{A}) + | \mbt{B}\setminus \{ S(\mbt{A}) \cup \mbt{A}\}| \ge \wh{K}( \mbt{B}), \mbox{ for all }  \mbt{A}\subseteq  \mbt{B}\subseteq  \mbt{X}  \mbox{ \textit{(slope*)}},
\end{align*}
respectively.
\end{lemma}

\begin{proof}
First we show that without loss of optimality we can add the  $\mbox{ \textit{initialize*}}$  and $\mbox{ \textit{slope*}}$ conditions to the LP in Table \ref{table2}. Since they are more stringent than $\mbox{\textit{initialize}}$   and $\mbox{\textit{slope}}$ conditions  in Table \ref{table2}, the result then follows. We begin with  $\mbox{\textit{initialize*}}$. Let $\wh{K}(\mbt{A})$, $\mbt{A} \subseteq \mbt{X}$ be feasible for the LP in Table \ref{table2} such that $\wh{K}(\mbt{X}) > |\mbt{X}|$. Then there exists $\epsilon > 0$ such that $\wh{K}(\mbt{X}) = |\mbt{X}| + \epsilon$. Note that $\wh{K}(\mbt{A}) - \epsilon$,  $\mbt{A} \subseteq \mbt{X}$, is also feasible for the LP  in Table \ref{table2} giving a lower objective $\wh{K}(\emptyset) - \epsilon$. Hence, without loss of optimality we can insert the $\mbox{\textit{initialize*}}$  condition into the LP  in Table \ref{table2}.
Now we show that the $\mbox{\textit{slope}}$  and $\mbox{\textit{decode}}$ conditions of the LP in Table \ref{table2}  imply the \textit{slope*} condition.
Let $\mbt{A} \subseteq \mbt{B} \subseteq \mbt{X}$. If $\mbt{B} \cap S(\mbt{A}) = \emptyset$ then  the $\mbox{\textit{slope}}$ and  $\mbox{\textit{slope*}}$ conditions are equivalent. Otherwise, i.e., if  $\mbt{B} \cap S(\mbt{A}) = \mbt{C} \neq \emptyset$, then
from the $\mbox{\textit{decode}}$ and $\mbox{\textit{slope}}$ conditions we have 
\begin{align*}
&\wh{K}(\mbt{C} \cup \mbt{A}) = \wh{K}(\mbt{A}),
\\
&\wh{K}(\mbt{C} \cup \mbt{A})  + | \mbt{B} \setminus \{\mbt{C} \cup \mbt{A}\} |\ge  \wh{K}(\mbt{B})
\end{align*}
respectively. Since $\mbt{B} \setminus \{\mbt{C} \cup \mbt{A}\} =  \mbt{B} \setminus \{S(\mbt{A}) \cup \mbt{A}\} $,  the $\mbox{\textit{decode}}$ and $\mbox{\textit{slope}}$ conditions imply the $\mbox{\textit{slope*}}$ condition.
\end{proof}
\begin{lemma}
\label{lemma:LP_eqv}
Let $\epsilon > 0$ and $\bar{R}^{LP}_{lb}(\epsilon\mb{1})$ be the optimal value of the LP in Table \ref{table3}.
Then $ R^I_{lb}(\epsilon\mb{1}) \ge \bar{R}^{LP}_{lb}(\epsilon\mb{1})$ and  $\lim_{\epsilon \rightarrow 0} R^I_{lb}(\epsilon\mb{1}) = \bar{R}^{LP}_{lb}(\mb{0})$.
\begin{table*}[t]
\caption{Relaxation of the LP in Table \ref{table1}}
\label{table3}
$\min K(\emptyset) - \epsilon \mbox{ \textit{subject to}} $
\\
$K( \mbt{X}) = 0  \mbox{  \textit{(initialize)}} $
\\
$K( \mbt{A}) \ge 0, \mbox{ for all }  \mbt{A} \subseteq  \mbt{X}  \mbox{  \textit{(non-negativity)}}$ 
\\
$K( \mbt{B}) + H( \mbt{B} |  \mbt{A}) \ge K( \mbt{A}), \mbox{ for all }  \mbt{A} \subseteq  \mbt{B} \subseteq   \mbt{X}  \mbox{ \textit{(slope)}}$
\\
$K( \mbt{A}) \ge K( \mbt{B}), \mbox{ for all }  \mbt{A} \subseteq  \mbt{B} \subseteq   \mbt{X}  \mbox{ \textit{(monotonicity)}}$
\\
 $K( \mbt{A})  \ge K( \mbt{B}) + H( \mbt{B} |  \mbt{A}) - H( \mbt{B} |  S(\mbt{A}),\mbt{A}) - \epsilon\log|S(\mbt{A})|, \mbox{ for all }   \mbt{A} \subseteq  \mbt{B} \subseteq  \mbt{X} \mbox{ \textit{(monotonicity+)}}$
\\
$K( \mbt{A}) + K( \mbt{B}) \ge K( \mbt{A} \cap  \mbt{B}) + K( \mbt{A} \cup  \mbt{B}),$
\\
$\mbox{ for all }   \mbt{A},  \mbt{B} \subseteq  \mbt{X}  \mbox{ \textit{(submodularity)}}$.
\end{table*}
\end{lemma}
\begin{proof}
Since the  random variables $ \mbt{A}, \mbt{B}$ in  $ R^I_{lb}(\epsilon\mb{1})$ are such that $ \mbt{A}, \mbt{B} \subseteq  \mbt{X}$, the Markov chain $\mbt{A} \lra \mbt{B} \lra \mbt{X}$ is equivalent to $\mbt{A} \subseteq \mbt{B} \subseteq \mbt{X}$.
Then the \textit{slope} constraints of the LP in $R^I_{lb}(\epsilon\mb{1})$ imply the \textit{slope} constraints of $\bar{R}^{LP}_{lb}(\epsilon\mb{1})$, since $H(\mbt{B} | \mbt{A}) \ge I(\mbt{B}; \mbt{V}, \mbt{U_B}|\mbt{A})$.
Furthermore, using Fano's inequality, it can be seen that the \textit{monotonicity+} condition of the LP in $R^I_{lb}(\epsilon\mb{1})$ gives the \textit{monotonicity+} condition of  $\bar{R}^{LP}_{lb}(\epsilon\mb{1})$ and
the rest of the conditions are the same. Hence, we have $R^I_{lb}(\epsilon\mb{1}) \ge \bar{R}^{LP}_{lb}(\epsilon\mb{1})$.
Now we select $\mbt{V} = \mbt{Z}$ where $\mbt{Z}$ is a vector of i.i.d.\ Bernoulli$(\frac{1}{2})$ bits of the same length as $\mbt{X}$, $\mbt{Z} \perp \mbt{X}$, and we select $\mbt{U_A} = (S(\mbt{A}), \mbt{A})\oplus \mbt{Z}$,\footnote{ $\mb{a} \oplus\mb{b}$ denotes componentwise exclusive-OR operation where the shorter vector is zero padded as necessary.} $\mbt{A} \subseteq  \mbt{X}$. Note that this selection of $\mbt{V}$ and $\mbt{U_A}$ satisfy the conditions 1)--3) in Theorem \ref{theorem:lowergen}. Then the solution of the resulting  LP is equal to the LP in Table \ref{table3} where $\epsilon\log|S(\mbt{A})| = 0$, giving $\bar{R}^{LP}_{lb}(\mb{0}) - \epsilon \ge R^I_{lb}(\epsilon\mb{1})$. Since $\bar{R}^{LP}_{lb}(\epsilon\mb{1})$ is right-continuous at $\epsilon = 0$ \cite{cont_LP}, letting $\epsilon \rightarrow 0$ gives the result. 
\end{proof}
\begin{theorem}
\label{theorem:subsume_index}
$\lim_{\epsilon \rightarrow 0}R^I_{lb}(\epsilon\mb{1}) = \wh{R}^{LP}_{lb}$.
\end{theorem}
\begin{proof}
Let  $LP_1$ and $LP_2$ denote the LPs in Theorem \ref{theorem:blasiak} and Table \ref{table3}  with $\epsilon = 0$, respectively. By Lemma \ref{lemma:LP_insert}, without loss of optimality we can add the \textit{initialize*} and \textit{slope*} conditions in Lemma \ref{lemma:LP_insert} to $LP_1$ and consider $LP_1$  of this form. Notice that $\bar{R}^{LP}_{lb}(\mb{0})$ is the solution of $LP_2$ and from Lemma \ref{lemma:LP_eqv},  $\lim_{\epsilon \rightarrow 0}R^I_{lb}(\epsilon\mb{1}) = \bar{R}^{LP}_{lb}(\mb{0})$.
Hence, it is enough to show that $\wh{R}^{LP}_{lb} = \bar{R}^{LP}_{lb}(\mb{0})$.
We show this by reparametrizing $LP_2$ in terms of $\wh{K}( \mbt{A})$ where $ \wh{K}( \mbt{A}) = K( \mbt{A}) + H( \mbt{A})$. Note that 
$\wh{K}(\emptyset) = K(\emptyset).$
Hence, the objective of $LP_2$ is the same as the objective of $LP_1$.
Now we show that the constraint set in $LP_2$ and the constraint set in $LP_1$ are the same.
We can rewrite the \textit{initialize} and \textit{non-negativity} conditions of $LP_2$ as
\\
$\wh{K}( \mbt{X})  = H( \mbt{X})$
\\
$\wh{K}( \mbt{A}) \ge H( \mbt{A})$ respectively. Together those two conditions are equivalent to the \textit{initialize*} and \textit{slope} conditions of $LP_1$.

When we rewrite the \textit{slope} condition of $LP_2$, we get
\\
$\wh{K}( \mbt{B})  \ge \wh{K}( \mbt{A})$,
the \textit{monotonicity} condition of $LP_1$.

When we rewrite the \textit{monotonicity} and \textit{monotonicity+} conditions of $LP_2$, we get
\\
$\wh{K}( \mbt{A}) +H( \mbt{B}| \mbt{A}) \ge \wh{K}( \mbt{B})$
\\
$\wh{K}( \mbt{A}) + H( \mbt{B}| S(\mbt{A}),\mbt{A}) \ge \wh{K}( \mbt{B})$ respectively and they are equivalent to the
\textit{slope} and \textit{slope*} conditions of $LP_1$.

Also, combining the \textit{submodularity} condition of $LP_2$ and $H( \mbt{A}) + H( \mbt{B}) = H( \mbt{B} \cap  \mbt{A}) + H( \mbt{B} \cup  \mbt{A})$ we can get the same \textit{submodularity} condition of $LP_1$.

Lastly,  from the \textit{monotonicity+} and \textit{slope}  conditions of $LP_2$, we can obtain
$K( \mbt{A}) + H( \mbt{A}) = K( \mbt{B}) + H( \mbt{B}| \mbt{A}) + H( \mbt{A})$ for all $A \leadsto B$, which is  the \textit{decode} condition of $LP_1$. 
Hence, each constraint (or combination of constraints) in $LP_2$ corresponds to a constraint in $LP_1$ and vice versa. Since the objectives of $LP_1$ and $LP_2$ are the same, we conclude that  $\wh{R}^{LP}_{lb} =  \bar{R}^{LP}_{lb}(\mb{0})$.
\end{proof}

\section{Optimality Results}
\label{sec:odd_cycle_gauss}
The LP upper and  lower bounds are tight in several instances\footnote{
In a recent work of Benammar \textit{et al.}~\cite{benammar}, the rate-distortion problem with two decoders having degraded reconstruction sets is considered and
the corresponding rate-distortion function is characterized. 
The construction of auxiliary random variables in the converse result of Benammar \textit{et al.} \cite{benammar} is specific to that problem setting and at this point it is unclear whether the LP lower bound subsumes this converse result.}. We begin with several classes of instances for which the rate-distortion function is
already known, the last of which is the odd-cycle index coding problem, which can be considered as a special case of Heegard-Berger problem. We conclude this section by finding an explicit characterization of the rate-distortion function for a new ``odd-cycle Gaussian  problem" using the upper and lower bounds in Theorems \ref{thm:gen_ach} and \ref{theorem:lowergen_comp}, respectively.

\subsection{Rate-Distortion Function with Mismatched Side Information at Decoders \cite{watanabe}}
In this problem, there is one encoder with source $ \mbt{X} = (X_1, X_2)$ and two decoders with side information $ \mbt{Y_1} = (Y_{11}, Y_{12})$ and $ \mbt{Y_2} = (Y_{21}, Y_{22})$, respectively. The source and side information satisfy the following relations
\begin{align}
\label{wt:indp}
& (X_1,Y_{11}, Y_{21}) \perp (X_2,Y_{12}, Y_{22})
\\
\label{wt:degraded}
&X_1\lra Y_{11} \lra Y_{21} \mbox{ and }
X_2\lra Y_{22} \lra Y_{12}
\end{align}
and the reconstructions at the decoders, $ \mbt{\wh{X}_1} = (\wh{X}_{11},\wh{X}_{12})$ and $ \mbt{\wh{X}_2} = (\wh{X}_{21}, \wh{X}_{22})$, are such that
\begin{align}
\label{wt:dist_const1}
&\mathbb{E}[d_{1i}(X_{1},\wh{X}_{1i})] \le D_{1i}
\\
\label{wt:dist_const2}
&\mathbb{E}[d_{2i}(X_{2},\wh{X}_{2i})] \le D_{2i}
\mbox{ for }  i \in [2].
\end{align}

We denote the rate-distortion function of this problem as $R^M{(\mb{D})}$. Theorem \ref{theorem:subs_mismatched} shows that the minimax lower bound in Theorem \ref{theorem:minmax} is greater than or equal to $R^M({\mb{D}})$, the rate-distortion function characterized by Watanabe \cite{watanabe}. Hence,  it implies 
that the lower bounds in both Theorems \ref{theorem:minmax} and \ref{theorem:lowergen} are tight for this problem.

\begin{theorem}[\cite{watanabe}]
\label{theorem:watanabe}
\label{theorem-main}
The rate-distortion function, $R^M(\mb{D})$, equals
\begin{align*}
 &R^M(\mb{D}) = 
\min[ \max\{  R^M_1, R^M_2 \} ], \mbox{ where }
\\
 &R^M_1 = I(X_1; W_1|Y_{11}) + I(X_2;W_2|Y_{12}) 
 + I(X_1;U_1|Y_{11},W_1)  + I(X_2;U_2|Y_{22},W_2)
 \\
  &R^M_2 =  I(X_1;W_1|Y_{21}) + I(X_2;W_2|Y_{22}) 
  + I(X_1;U_1|Y_{11},W_1) + I(X_2;U_2|Y_{22},W_2) ,
\end{align*}
and the minimization is taken over all auxiliary random variables
$W_1,W_2,U_1,U_2$ satisfying the following:
\\
1) $(W_i,U_i) \leftrightarrow X_i \leftrightarrow (Y_{1i}, Y_{2i})$ for $i=1,2$.
\\
2) $(W_1,U_1,X_1,Y_{11},Y_{21})$ and $(W_2,U_2,X_2,Y_{12},Y_{22})$ are independent of each other.
\\
3) There exist functions $g_{11}(W_1,U_1,Y_{11}) = \wh{X}_{11}$,
$g_{12}(W_2,Y_{12}) = \wh{X}_{12}$, $g_{21}(W_1,Y_{21}) = \wh{X}_{21}$,
and $g_{22}(W_2,U_2,Y_{22}) = \wh{X}_{22}$ such that they satisfy (\ref{wt:dist_const1}) and (\ref{wt:dist_const2}). 
\\
4) $|{\cal W}_i| \le |{\cal X}_i| + 3$ and 
$|{\cal U}_i| \le |{\cal X}_i| \cdot (|{\cal X}_i| + 3) + 1$ for $i=1,2$, where
${\cal W}_i$ and ${\cal U}_i$ are alphabets of $W_i$ and $U_i$
respectively.
\end{theorem}

\begin{theorem}
\label{theorem:subs_mismatched}
$\liminf_{\epsilon \rightarrow 0}R^m_{lb}(\mb{D} + \epsilon \mb{1}) \ge R^M(\mb{D})$  and $R_{ach}(\mb{D}) \le R^M(\mb{D})$.
\end{theorem}

\begin{proof}
We select the joint distribution of $( \mbt{X}, \mbt{Y_1}, \mbt{Y_2})$ such that it satisfies (\ref{wt:degraded}). 
First we show $\liminf_{\epsilon \rightarrow 0}R^m_{lb}(\mb{D} + \epsilon \mb{1}) \ge R^M(\mb{D})$.
Let $U_{Y} = (V, U_{Y_1})$ and $U_{Z} =(V, U_{Y_2})$. Then $\bar{R}_{lb}$ in Theorem \ref{theorem:minmax}  can be written as\footnote{Note that Theorem \ref{theorem:minmax} can be applied to vector-valued sources and side information at the decoders.}
$\bar{R}_{lb} = \max \{\bar{R}_{lb1}, \bar{R}_{lb2}\}$, 
\begin{align*}
&\bar{R}_{lb1} =  I( \mbt{X}; U_{Y}| \mbt{Y_1}) + I( \mbt{X}; U_{Z}| U_{Y}, \mbt{Y_{1}}, \mbt{Y_{2}})
\\
& \bar{R}_{lb2} =   I( \mbt{X};  U_{Z}| \mbt{Y_2}) + I( \mbt{X}; U_{Y}| U_{Z}, \mbt{Y_{1}}, \mbt{Y_{2}}).
\end{align*}

By the chain rule and using (\ref{wt:degraded}), $R_{lb1}$ can be rewritten as
\begin{align*}
&I(X_2;U_{Y},Y_{11}|Y_{12}) + I(X_1;U_{Y},Y_{12},X_2|Y_{11}) 
 + I(X_2;U_{Z}|U_{Y},Y_{11},Y_{22}) +  I(X_1;U_{Z}|U_{Y},Y_{22},X_2,Y_{11})
\\
&\overset{a}{=}I(X_2;U_{Y},Y_{11}|Y_{12}) + I(X_1;U_{Y},Y_{22},X_2|Y_{11})
  + I(X_2;U_{Z}|U_{Y},Y_{11},Y_{22}) +  I(X_1;U_{Z}|U_{Y},Y_{22},X_2,Y_{11})
%
\\
&\overset{b}{=} I(X_2;U_{Y},Y_{11}|Y_{12}) + I(X_1;U_{Y},Y_{22},X_2,U_{Z}|Y_{11}) 
+ I(X_2;U_{Z}|U_{Y},Y_{11},Y_{22}) 
\\
&\ge I(X_2;U_{Y},Y_{11}|Y_{12}) + I(X_1;U_{Y},Y_{22},U_{Z}|Y_{11}) 
 + I(X_2;U_{Z}|U_{Y},Y_{11},Y_{22}),
\end{align*}
 which equals $I(X_2;W_2|Y_{12}) + I(X_1;W_1,U_1|Y_{11}) + I(X_2;U_2|W_2,Y_{22})$ $= R^M_1$,
where $W_2= (V, U_{Y_1},Y_{11})$, $W_1 = (V, U_{Y_2},Y_{22})$,
$U_1 = U_{Y_1}$ and $U_2 = U_{Y_2}$. Here
\\
a: follows since $I(X_1;Y_{12},Y_{22}| X_2,U_{Y},U_{Z},Y_{11}) =0$.
\\
b: follows by combining the second and last term.

Similarly, we can obtain $\bar{R}_{lb2} \ge R^M_2$. 

Note that $(U_{Y},U_{Z})\lra (X_1,X_2) \lra (Y_{11},Y_{12},Y_{21},Y_{22})$ implies the first condition of the minimization in Theorem \ref{theorem:watanabe}.
Also, the distortion constraints in $R^m_{lb}(\mb{D} + \epsilon\mb{1})$ imply the third 
condition of the minimization with $\epsilon$ added to distortion constraints in Theorem \ref{theorem:watanabe}.
Hence, we can write 
\begin{align}
\label{ineq:proof_watanabe}
R^m_{lb}(\mb{D} + \epsilon\mb{1}) \ge \inf[\max\{ R^M_1,  R^M_2\}]  - \epsilon, 
\end{align}
where the minimization is over $(W_1, U_1,W_2, U_2)$ satisfying the first and third conditions of the minimization  in Theorem \ref{theorem:watanabe}.
Also, since $(W_1, U_1)$ and $(W_2, U_2)$ do not appear together, we can add the condition 2) in Theorem \ref{theorem:watanabe} to the minimization in (\ref{ineq:proof_watanabe}). Lastly, cardinality bounds on $(W_1,W_2,U_1,U_2)$ can be obtained as in $R^M(\mb{D})$  and  the right-hand side of (\ref{ineq:proof_watanabe}) can be shown to be continuous in $\epsilon$ using the same procedure as in \cite{watanabe}.

It remains to show that $R_{ach}(\mb{D}) \le R^M(\mb{D})$. In \cite{watanabe}, $R^T(\mb{D})$ in Lemma \ref{lemma:ach_2dec} is utilized to obtain $R^M(\mb{D})$. Hence, from Lemma \ref{lemma:ach_2dec}, we have $R_{ach}(\mb{D}) \le R^M(\mb{D})$.
\end{proof}
\subsection{Rate-Distortion Function with Conditionally Less Noisy Side Information \cite{timo_lessnoisy}} 
There are two decoders, and the distortion measure at decoder 1, 
$d_1(\cdot,\cdot)$, is such that $d_1(X,\wh{X}) = 0$ if $\wh{X} = a(X)$ and $d_1(X,\wh{X}) = 1$ otherwise, where $a(X)$ is a deterministic map. Also the allowable distortion at decoder $1$, $D_1$, is taken as zero. Timo \textit{et al.} \cite{timo_lessnoisy} show that their lower bound for this problem is tight if $Y_2$ is conditionally less noisy than $Y_1$, i.e., $(Y_2 \succeq Y_1 | a(X))$, and $ H(a(X) | Y_1)  \ge  H(a(X) | Y_2)$. Although whether the minimax lower bound in Theorem \ref{theorem:minmax} is tight for this problem is not known, the next theorem shows that $R_{lb}(\mb{D} + \epsilon \mb{1})$ subsumes the lower bound  in \cite{timo_lessnoisy}  when  $ (Y_2 \succeq Y_1 | a(X))$.
\begin{theorem}
$\liminf_{\epsilon \rightarrow 0}R_{lb}(\mb{D} + \epsilon \mb{1})  \ge R^{LN}(\mb{D})$  and $R_{ach}(\mb{D}) \le R^{LN}(\mb{D})$ where 
\begin{align*}
R^{LN}(\mb{D}) =& H(a(X) | Y_1) 
 + \min_{
\substack{
W\leftrightarrow X \leftrightarrow (a(X), Y_2)
\\
\mathbb{E[} d_2(X, g_2(W, a(X),Y_2) )]\le D_2 
\\ |\mathcal{W}| \le |\mathcal{X}| + 1,
}} I(X; W | a(X), Y_2)
\end{align*}
is the lower bound in \cite[Lemma 5]{timo_lessnoisy} when $ (Y_2 \succeq Y_1 | a(X))$.
\end{theorem}

\begin{proof}
We begin with showing $\liminf_{\epsilon \rightarrow 0}R_{lb}(\mb{D} + \epsilon \mb{1}) \ge R^{LN}(\mb{D})$. Similar to the proof of Theorem \ref{theorem:minmax}, first we consider  $R_{lb}(\mb{D} + \epsilon \mb{1})$. 
For a given $\epsilon > 0$ we can write
\begin{align}
\label{ineq:relaxedlb2}
R_{lb}(\mb{D} + \epsilon \mb{1})  \ge \inf_{V \in \mathcal{C}(X)} \inf_{U_{\cdot}  : \mathcal{C}(X) \rightarrow \mathcal{C}(X, V)} {R}^{LP}_{lb}(\epsilon) 
\end{align}
where the $LP$ constraints on the random variables $(X, a(X),Y_1,Y_2)$ are as in the problem description. Now we find a lower bound to 
$R^{LP}_{lb}(\epsilon)$ in (\ref{ineq:relaxedlb2}) by utilizing some of the LP constraints.  Note that we can write
\begin{align*}
&K(\emptyset) \ge K(Y_1) \mbox{ by \textit{(monotonicity)}}
\\
&K(Y_1)  \ge  K(a(X), Y_1) + H(a(X) | Y_1) - \delta(\epsilon)
\mbox{ by  \textit{(monotonicity+)}, Fano's inequality, and  }\delta(\epsilon) > 0, 
\\
&K(a(X), Y_1) \ge K(a(X), Y_2) \mbox{ by \textit{(monotonicity)}}, 
\\
&K(a(X), Y_2) \ge I(X; V,U_{a(X)Y_2} | a(X), Y_2 )
\mbox{ by \textit{(monotonicity+)} and }
K(X, a(X), Y_2) = 0.
\end{align*}
 Hence, $R_{lb}(\mb{D}+ \epsilon\mb{1})$ is lower bounded by 
\begin{align*}
\inf_{V \in \mathcal{C}(X)} \inf_{U_{\cdot}  : \mathcal{C}(X) \rightarrow \mathcal{C}(X, V)}  H(a(X) | Y_1)  +  I(X; V, U_{a(X)Y_2} | a(X), Y_2 ) -\delta(\epsilon).
\end{align*}
By finding a cardinality constraint on $(V, U_{a(X)Y_2})$ and  letting $\epsilon \rightarrow 0$,
we have the result.

Now we show that $R_{ach}(\mb{D}) \le R^{LN}(\mb{D})$. By selecting the auxiliary random variables $U_{\{1,2\}} = a(X)$, $U_{\{1\}} = \emptyset$ and $U_{\{2\}} = W$ in Lemma \ref{lemma:ach_2dec} and imposing the cardinality constraint $|\mathcal{W}| \le |\mathcal{X}| +1$, we have $R_{ach}(\mb{D}) \le R^{LN}(\mb{D})$.
\end{proof}

\subsection{Odd-cycle Index Coding Problem}
 The source $\mb{X} = (X_1, \ldots, X_m)$, where $m \ge 5$ is an odd number, is i.i.d.\ Bernoulli ($\frac{1}{2}$) bits. The side information at decoder $i$, $i \in [m]$ is $\mb{Y_i} = (X_{i-1}, X_{i+1})$, where $+$ and $-$ in subscripts are modulo-$m$ operations\footnote{Here $x \mod m$ is assumed to lie in $[m]$.}, and decoder $i$ wishes to reconstruct $X_i$ with a vanishing block error probability.

Although the achievability result Theorem \ref{thm:gen_ach} is for  per-letter distortion constraints,  it can be easily modified to accommodate block error probabilities. Let $v \in \set{V}$ be fixed. Then we select the messages
 $\mbt{U}_{\set{S}_j}$, $\set{S}_j \in v$ such that 
\begin{align}
\mbt{U}_{jk} = (X_{j}, X_{k}) \mbox{ for } j \in [m],  k \equiv j+1\mod m
\end{align}
and all of the other messages $\mbt{U_{\set{S}_j}}$ are chosen to be $\emptyset$. \footnote{We represent $\mbt{U}_{\{j,k\}}$ as $\mbt{U}_{jk}$ for ease of notation.} 
 Let $j \in [m]$,  $i \equiv j-1\mod m$ , $k \equiv j+1 \mod m$, and $l \equiv k+1\mod m$. Then from the conditions in $C^{LP}_{ach}$,  we can write
\begin{align}
R_{jk} &\ge I\big(\mb{X};\mbt{U}_{jk}\big) - R'_{jk}, \mbox{ from condition 2) of } C^{LP}_{ach} 
 \notag  \\
& = 2 - R'_{jk}, 
  \notag  \\
 R'_{jk} &\le
\min\{ I\left(\mbt{U}_{jk};\mbt{U}_{ij},\mb{Y_j}\right),  I\left( \mbt{U}_{jk}; \mbt{U}_{kl},\mb{Y_k}\right)\}, \mbox{ from condition 3) of } C^{LP}_{ach} 
\notag \\
&= 2,
\notag \\
 R'_{ij} + R'_{jk} &\le
 I\left(\mbt{U}_{ij}; \mb{Y_j}\right) +  I\left(\mbt{U}_{jk}; \mbt{U}_{ij},\mb{Y_j}\right),  \mbox{ from condition 3) of } C^{LP}_{ach}
 \notag \\
&= 3.
\end{align}

Then selecting  $R'_{jk} =  \frac{3}{2}$ and $R_{jk} = \frac{1}{2}$
 satisfies the conditions of $C^{LP}_{ach}$. Hence,  rate  $\frac{m}{2}$ is achievable. Also, in  \cite{blasiak}  it is shown that the LP lower bound in Theorem \ref{theorem:blasiak} gives $\frac{m}{2}$ for the zero error case. From Theorem \ref{theorem:subsume_index}, we can conclude that the $R^I_{lb}$ lower bound, which is a valid lower bound for vanishing error probability, also gives  $\frac{m}{2}$ which is the optimal rate for this problem.

Note that prior to this work, the minimax lower bound in Theorem \ref{theorem:minmax}, $R^m_{lb}(\mb{D}+ \epsilon\mb{1})$, was the state-of-the-art lower bound to $R(\mb{D})$ for multiple decoders. The next lemma states that the minimax lower bound is strictly suboptimal for the odd-cycle index coding problem.
\begin{lemma}
  $\limsup_{\epsilon \rightarrow 0}R^m_{lb}(\mb{0}+ \epsilon\mb{1}) < \frac{m}{2}$ for the odd-cycle index coding problem.
\end{lemma}
\begin{proof}
Firstly, note that $\limsup_{\epsilon \rightarrow 0} R^m_{lb}(\mb{0}+ \epsilon\mb{1}) \le \frac{m}{2}$ since $\frac{m}{2}$ is the optimal rate for the odd-cycle case.  Also, notice that when we select $(V, U_{Y_1}, \ldots, U_{Y_m})$ such that $V = \emptyset$, and $U_{Y_{i}} = X_i$, $i \in [m]$,  the random variables $(V, U_{Y_1}, \ldots, U_{Y_m})$ are feasible in the optimization problem in Theorem \ref{theorem:minmax} and  $\bar{R}_{lb}$ becomes 
\begin{align}
\label{eq:odd_cycle_minimax_int}
\max_{\sigma} \big{[} H( X_{\sigma(1)}|Y_{\sigma(1)})
+ H(X_{\sigma(2)}|  X_{\sigma(1)},Y_{\sigma(1)},Y_{\sigma(2)}) +\cdots
+ H( X_{\sigma(m)}| X_{\sigma(1)},\ldots, X_{\sigma(m-1)},{Y}) \big{]},
\end{align}
which is equal to the maximin lower bound for index coding in \cite{sinem_index} implying $\liminf_{\epsilon \rightarrow 0} R^m_{lb}(\mb{0}+ \epsilon\mb{1}) = \bar{R}_{lb}$. The fact that (\ref{eq:odd_cycle_minimax_int}) must take a integer value concludes the proof.
\end{proof}
Note that if we restrict the selection of auxiliary random variables to be a subset of the source $\mbt{X}$ in Theorem 2 of \cite{timo}, the scheme in \cite{timo} becomes valid and each mutual information term in the optimization becomes an entropy of a subset of the source which gives an integer value. Hence, in this case the scheme in \cite{timo} gives an integer rate which is strictly suboptimal for this problem.
\subsection{Odd-cycle Gaussian Rate-Distortion Problem}

We finish with an instance that seems not to be solvable
using existing lower bounds discussed in Section \ref{sec:subsumes}.
The problem setting we consider is analogous to the odd-cycle index coding problem \cite{blasiak}, by taking each source component as an independent Gaussian random variable instead of uniform binary bits and considering a mean square error (MSE) distortion constraint on the reconstructions. Hence, we call it the \textit{odd-cycle Gaussian} problem from now on. Specifically, the source $\mb{X} = (X_1, \ldots, X_m)$, where $m \ge 5$ is an odd number, is a Gaussian vector such that each component is independent of the others and has unit variance.  The side information at decoder $i$, $i \in [m]$ is $\mb{Y_i} = (X_{i-1}, X_{i+1})$, where $+$ and $-$ in the subscripts are modulo-$m$ operations\footnote{Here as well, $x \mod m$ is assumed to lie in $[m]$.}, and decoder $i$ wishes to reconstruct $X_i$ subject to an MSE distortion constraint, i.e.,
$E[(X_i - \wh{X}_i)^2] \le D$ for all $i \in [m]$.
\begin{theorem}
\label{thm:odd_cycle_GI}
The rate-distortion function, $R^{IG}(\mb{D})$, is
\begin{align}
\label{RD:odd_cycle_GI}
&R^{IG}(\mb{D}) = \frac{m}{4}\log\frac{1}{D}.
\end{align} 
\end{theorem}
\begin{proof}[Proof of Theorem \ref{thm:odd_cycle_GI}]
\textit{Achievability:} The achievability argument is obtained by using Theorem \ref{thm:gen_ach_gauss}. Let $v \in \set{V}$  be fixed. We select the messages $U_{\set{S}_j}$ such that 
\begin{align}
\mb{U}_{jk} = (X_{j} + N_{j}, X_{k} +\bar{N}_{k}) \mbox{ for } j \in [m],  k \equiv j+1\mod m
\end{align}
and all the other messages $U_{\set{S}_j}$ are degenerate.\footnote{We represent $U_{\{j,k\}}$ as $U_{jk}$ for ease of notation.} 
Here $(N_i, \bar{N}_i)$, $i \in [m]$  are Gaussian random variables with variance $K_{N_i} = K_{\bar{N}_i}=\frac{2D}{1-D} $ and all $N_i, \bar{N}_i$'s are independent of each other and the source $\mb{X}$. All $\mb{U}_{\set{S}_j}$ satisfy conditions 1), 2)  and 3) of $C^G_{ach,v}(\mb{D})$ as well as condition 4) of $C^G_{ach,v}(\mb{D})$ since $K_{X_j|\mb{U}_{jk}, \mb{U}_{ij},\mb{Y_j}} = (K^{-1}_{X_j} + K^{-1}_{N_j}+  K^{-1}_{\bar{N}_j})^{-1}=D $, where $i = j-1\mod m$.  Let $j \in [m]$,  $i \equiv j-1\mod m$, $k \equiv j+1 \mod m$, and $l \equiv k+1\mod m$.  Then from the conditions in $C^{LP}_{ach}$,  we can write
\begin{align}
\label{ineq:ach_G_lp1}
&R_{jk} \ge I\big(\mb{X};\mb{U}_{jk}\big) - R'_{jk}, \mbox{ from condition 2) of } C^{LP}_{ach} 
\\
&\mbox{ and since any disjoint sets of $\mb{U}_{\set{S}_j}$ are conditionally independent of each other given $\mb{X}$. }
 \notag  \\
 \label{ineq:ach_G_lp2}
& R'_{jk} \le
\min\{ I\left( \mb{U}_{jk}; \mb{U}_{ij},\mb{Y_j}\right),  I\left( \mb{U}_{jk}; \mb{U}_{kl},\mb{Y_k}\right)\}, 
 \\
 \label{ineq:ach_G_lp3}
& \mbox{ from condition 3) of } C^{LP}_{ach}.
\notag \\
& R'_{ij} + R'_{jk} \le
 I\left( \mb{U}_{ij}; \mb{Y_j}\right) +  I\left( \mb{U}_{jk};  \mb{U}_{ij},\mb{Y_j}\right), 
 \mbox{ by condition 3) of } C^{LP}_{ach}.
\end{align}

Note that the terms inside the minimum in (\ref{ineq:ach_G_lp2}) are equal to each other and also the encoding order of the messages does not affect  the right-hand side of (\ref{ineq:ach_G_lp3}). Then using the chain rule, the mutual information terms in (\ref{ineq:ach_G_lp1})--(\ref{ineq:ach_G_lp3}) can be written as
\begin{align*}
 I(\mb{X};\mb{U}_{jk}) &= I(X_j; X_j + N_j) +I(X_k; X_k + \bar{N}_k)
\\
&= \log\frac{1+D}{2D}.
\\
I\left( \mb{U}_{jk}; \mb{U}_{ij},\mb{Y_j}\right) &=I\left( \mb{U}_{jk}; \mb{U}_{ij},X_i,X_k\right)
\\
&= I(X_k + \bar{N}_k;X_k)+I\left( X_j + N_j; X_j + \bar{N}_j \right) 
\\
& =  \frac{1}{2}\log\frac{1+D}{2D} + \frac{1}{2}\log\frac{(1+D)^2}{4D}.
\\
 I\left( \mb{U}_{ij}; \mb{Y_j}\right) +  I\left( \mb{U}_{jk};  \mb{U}_{ij},\mb{Y_j}\right)
 &=  I\left( \mb{U}_{ij}; X_i,X_k\right) +  I\left( \mb{U}_{jk}; \mb{U}_{ij},\mb{Y_j}\right) 
 \\
 &= \frac{1}{2}\log\frac{1+D}{2D} + \frac{1}{2}\log\frac{1+D}{2D} + \frac{1}{2}\log\frac{(1+D)^2}{4D}.
\end{align*}

Then selecting  $R'_{jk} =  \frac{1}{2}\log\frac{1+D}{2D} + \frac{1}{4}\log\frac{(1+D)^2}{4D}$ and $R_{jk} =  \log\frac{1+D}{2D} - R'_{jk}$, $j \in [m]$, $k = j+1 \mod m$ satisfies (\ref{ineq:ach_G_lp1})--(\ref{ineq:ach_G_lp3}) and we take all other rates $R_{\set{S}_j},R'_{\set{S}_j}$ as $0$. Hence, the achievable rate is
\begin{align*}
\sum^m_{i =1}R_{ij} &= m\left( \log\frac{1+D}{2D}  - \frac{1}{2}\log\frac{1+D}{2D} - \frac{1}{4}\log\frac{(1+D)^2}{4D}\right)
\\
& =\frac{m}{4}\left(2\log\frac{1+D}{2D} -\log\frac{(1+D)^2}{4D}\right)
\\
& = \frac{m}{4}\log\frac{1}{D}.
\end{align*}
\\
\textit{Converse:} We utilize the computable relaxation of $R_{lb}(\mb{D}+ \epsilon\mb{1})$ in Theorem \ref{theorem:lowergen_comp}. Similar to the proof of \cite[Theorem 5.1]{blasiak} we define the ordered sets:
\\
$\mb{O}=\{ X_i : i \equiv 1\mod 2, i\neq m \}$,
$\mb{O}^+ = \{ X_i : i \le m-2 \}$
\\
$\mb{E}=\{ X_i : i \equiv 0 \mod 2\}$,
$\mb{E}^+ = \{ X_i : 2\le i \le m-1\}$
\\
$\mb{M} =\{ X_i : 2 \le i \le m-2 \}$,
and $\mbt{S} = \mbt{X} \setminus (\mbt{M} \cup X_m)$. Note that $(\mbt{O}^+ \setminus \mbt{O}) \cap (\mbt{E}^+ \setminus \mbt{E}) = \emptyset $ and $\mbt{M} = (\mbt{O}^+ \setminus \mbt{O}) \cup (\mbt{E}^+ \setminus \mbt{E})$. Also, define  $R(D) =\frac{1}{2}\log\frac{1}{D}$.
Then using the conditions of the LP in Theorem \ref{theorem:lowergen_comp} we can obtain the following inequalities
\begin{align}
\label{ineq:conv:odd_cycle_GI_s}
&K(\emptyset) \ge K(\mbt{O}) 
\quad   \mbox{  by  (\textit{monotonicity})}
\\
&K(\emptyset) \ge K(\mbt{E}) \quad  \mbox{ by 
 (\textit{monotonicity})}
\\
&K(\emptyset) \ge K(X_m) \quad   \mbox{  by  (\textit{monotonicity})}
\\
\label{ineq:conv:odd_cycle_GI_e3}
&K(\mbt{O}) \ge K(\mbt{O}^+) +  \sum_{X_i \in
\mbt{O}^+ \setminus \mbt{O} }R(D + \epsilon) 
\end{align}
\begin{align}
\label{ineq:conv:odd_cycle_GI_e4}
&K(\mbt{E}) \ge K(\mbt{E}^+ ) +  \sum_{X_i \in
\mbt{E}^+ \setminus \mbt{E} }R(D + \epsilon)  
\\
\label{ineq:conv:odd_cycle_GI_e1}
&K(\mbt{O}^+) + K(\mbt{E}^+) \ge K(\mbt{M}) + K(\mbt{X}) + R(D + \epsilon)
\\
\label{ineq:conv:odd_cycle_GI_e2}
&K(\mb{M}) + K(X_m) \ge K(\emptyset) + K(\mb{X}) + \sum_{X_i \in \mb{S}}R(D + \epsilon)
\end{align}
where (\ref{ineq:conv:odd_cycle_GI_e3}) is due to the following.  By \textit{monotonicity+}, we have $K(\mbt{O}) \ge K(\mbt{O}^+) +  R(D_{\mbt{O}} + \epsilon)$. We can see that
$\sum_{X_i \in
\mbt{O}^+ \setminus \mbt{O} }R(D + \epsilon)$ is an upper bound to $R(D_{\mbt{O}} + \epsilon)$ by selecting the auxiliary random variable $V = \{X_i + N_i | i \in [m] \} $, where $N_i$ is independent of $\mb{X}$ and all $N_j$'s, $j \neq i$, in the minimization in Notation \ref{notation:computable}. Also, utilizing the chain rule one can verify that $R(D_{\mbt{O}} + \epsilon) \ge \sum_{X_i \in
\mbt{O}^+ \setminus \mbt{O} }R(D + \epsilon)$. By following a similar procedure to that used to obtain  (\ref{ineq:conv:odd_cycle_GI_e3}), we can also obtain
(\ref{ineq:conv:odd_cycle_GI_e4}). Furthermore, (\ref{ineq:conv:odd_cycle_GI_e1}) and (\ref{ineq:conv:odd_cycle_GI_e2}) are due to \textit{submodularity} and \textit{monotonicity+}. If we add inequalities (\ref{ineq:conv:odd_cycle_GI_s})--(\ref{ineq:conv:odd_cycle_GI_e2}) side-by-side, we obtain
$2K(\emptyset) \ge mR(D + \epsilon)$.
Taking $\epsilon \rightarrow 0$ gives the result.
\end{proof}
Recall that prior to the LP lower bound introduced here, the lower bound in Theorem \ref{theorem:minmax}, $R^m_{lb}(\mb{D}+ \epsilon\mb{1})$, was the state-of-the-art lower bound to $R(\mb{D})$. The next lemma shows that $R^m_{lb}(\mb{D}+ \epsilon\mb{1})$ gives $\frac{m-1}{4}\log(\frac{1}{D})$ for the odd-cycle Gaussian problem and is thus not tight.

\begin{lemma}
\label{lemma:minimax_odd_cycle_gauss}
For the odd-cycle Gaussian problem, 
$\liminf_{\epsilon \rightarrow 0} R^m_{lb}(\mb{D}+ \epsilon\mb{1}) = \frac{m-1}{4}\log(\frac{1}{D})$.
\end{lemma}

\begin{proof}[Proof of Lemma \ref{lemma:minimax_odd_cycle_gauss}]
The proof is given in the Appendix \ref{app:strict_ineq}.
\end{proof}

\section*{Acknowledgment}

The authors wish to thank Robert Kleinberg for several helpful comments.
This work was supported by the US National Science Foundation under
grant CCF-1617673.

\begin{appendices}

\section{}
 \label{app:gen_ach}
\begin{proof}[Proof of Theorem \ref{thm:gen_ach}]
Let $\epsilon > 0$, $v \in \set{V}$ be given and suppose the joint distribution of $(\set{U}, X, Y_1, \ldots, Y_m)$ in $C_{ach,v}(\mb{D})$, denoted by $p$, is fixed. The scheme consists of three main steps; namely, code construction, encoding and decoding. First we explain each step then show that the resulting rate  is $\mb{D}$-achievable.

Code construction and encoding are similar to the proof of the achievable scheme in \cite{timo}, which depends on $\epsilon$-letter typicality \cite{kramer} arguments. Here we use the lowercase letter $z$ to denote a realization of a random variable $Z$.
\\
\textit{Code Construction :}
 A codebook, denoted by $\set{C}^{\set{S}_j}$, of size $2^{n(R_{\set{S}_j} +R'_{\set{S}_j} )}$ is created for each set $\set{S}_j \in v$ in the following way. 
Let $\mb{k}_{\set{S}_j} = (k_{\set{S}_j}, k'_{\set{S}_j}) $, where  $k_{\set{S}_j} \in [2^{nR_{\set{S}_j}}]$ and $k'_{\set{S}_j}
\in [2^{nR'_{\set{S}_j}}]$.  A codeword $u_{\set{S}_j}(\mb{k}_{\set{S}_j}) \in \mathcal{U}_{\set{S}_j}^n$ of length $n$ is created by drawing each component from $\mathcal{U}_{\set{S}_j}$ with respect to $p(u_{\set{S}_j})$ in an i.i.d.\ way. 
\\
\textit{Encoding :} Let $0 < \epsilon_0 < \cdots < \epsilon_{2^m +1}$ be sufficiently small  and $x^n \in \mathcal{X}^n$ be given to the encoder. Then encoding is performed in $2^m-1$ stages. Specifically, at stage $j$ encoder picks $\set{C}^{\set{S}_j}$  and  searches for an index  $\mb{k}_{\set{S}_j}$ 
such that $u_{\set{S}_j}(\mb{k}_{\set{S}_j})$ is $\epsilon_j$-letter typical with $x^n$ and
\begin{align}
\label{eq:enc_typ1}
u_{\set{S}_j}^{-} &= \big\{u_{\set{S}_i}(\mb{k}_{\set{S}_i}) | i < j \big\}.
\end{align}
If such a $\mb{k}_{\set{S}_j}$ (or multiple such $\mb{k}_{\set{S}_j}$) exists then the encoder picks one of them arbitrarily and sends the bin index $k_{\set{S}_j}$ to the decoders. Otherwise the encoder picks a codeword randomly and sends the corresponding bin index.
\\
\textit{Decoding :} We apply simultaneous decoding \cite[Section 4]{ElGamal}. Consider decoder $l$. It  forms reconstructions of all its messages,  $u_{\set{D}_l}(\mb{\wh{k}}_{\set{D}_l}) =  \{u_{\set{S}_j}(\mb{\wh{k}}_{\set{S}_j}) | \set{S}_j \in \set{D}_l \}$, where $\mb{\wh{k}}_{\set{D}_l} = \{\mb{\wh{k}}_{\set{S}_j} | \set{S}_j \in \set{D}_l \}$ \footnote{ Since $v \in \set{V}$ is an ordered list, it induces an order on sets $\set{S}_j$. Hence we can take $\mb{\wh{k}}_{\set{D}_l}$ as an ordered set and assume an ordered set structure.}, in the following way.
Decoder $l$ takes the set of bin indices $k_{\set{D}_l}=\{k_{\set{S}_j} | \set{S}_j \in \set{D}_l \}$ then looks for a set of indices $\mb{\wt{k}}_{\set{D}_l}$ such that 
\begin{align}
\label{dec:cond1}
& \wt{k}_{\set{S}_j} = k_{\set{S}_j} \mbox{ for all } \set{S}_j \in \set{D}_l \mbox{ and }
\\
\label{dec:cond2}
&u_{\set{D}_l}(\mb{\wt{k}}_{\set{D}_l}) \mbox{ are $\epsilon_{l^*+1}$-letter typical with  } y_l^n,
\end{align}
where $l^* = \max_{j:\set{S}_j \in \set{D}_l} j$.
Note that if no error occurs at the encoder,  $u_{\set{D}_l}(\mb{k}_{\set{D}_l})$ is $\epsilon_{l^*}$-typical with $x^n$.
If there is more than one set of codewords $u_{\set{S}_j}(\mb{\wt{k}}_{\set{S}_j})$, $\set{S}_j \in \set{D}_l$ whose indices, $\mb{\wt{k}}_{\set{S}_j}$, satisfy (\ref{dec:cond1}) and (\ref{dec:cond2}), decoder $l$ selects one arbitrarily and sets $\mb{\wh{k}}_{\set{S}_j} = \mb{\wt{k}}_{\set{S}_j}$. If decoder $l$ cannot find any such set of indices, it sets $\mb{\wh{k}}_{\set{D}_l}$ to $\mb{1}$ (i.e., it declares an error). Since the joint distribution of $(\set{U},X,Y_1,\ldots,Y_m)$ is in $C_{ach,v}(\mb{D})$, we can find a function $g_l(\cdot,\cdot)$ such that $g_l(u_{\set{D}_l i}(\mb{\wh{k}}_{\set{D}_l}), y_{li}) = \hat{x}_{li}$, where $u_{\set{D}_l i}(\mb{\wh{k}}_{\set{D}_l})$, $y_{li}$ and $\hat{x}_{li}$ are the $i^{th}$ components of  $u_{\set{D}_l}(\mb{\wh{k}}_{\set{D}_l})$, $y^n_{l}$ and $\hat{x}^n_{l}$, respectively.

Now we analyze the error probabilities at the encoding and decoding steps, respectively.
\\
\textit{Error Analysis for Encoder :}
Note that encoding process is correct if the following is satisfied:
\\
1. At each encoding stage $j$, we can find  $U_{\set{S}_i}(\mb{k}_{\set{S}_i})$ such that it is $\epsilon_j$-jointly typical with $(U_{\set{S}_j}^{-}, X^n)$  i.e.,
\begin{align}
\label{eq:prob_cor_e1}
&C_{ \set{S}_j} = \left\{ \exists  \mb{ k}_{\set{S}_j} \mbox{ such that } u_{\set{S}_j}(\mb{k}_{\set{S}_j}) \in \mathcal{T}^{(n)}_{\epsilon_j}(p| U_{\set{S}_j}^{-}, X^n) \right\}.
\end{align}

Then the probability of error at the encoder, $\Pr(E)$ can be expressed as
\begin{align}
\label{eq:enc_error}
\Pr(E) &= \Pr((C_{\set{S}_1} \cap \ldots \cap C_{\set{S}_{2^m-1}})^c)
\notag \\
& = \Pr( C^c_{ \set{S}_1} \cup \ldots \cup C^c_{\set{S}_{2^m-1}})
\notag \\
& = \Pr((C^c_{\set{S}_1} \cap \bar{C}^1)\cup \ldots \cup (C^c_{\set{S}_{2^m-1}} \cap \bar{C}^{2^m-1})),
\end{align}
where $\bar{C}^{j}$ is defined as  $\bigcap^{i<j}_{i=1}C_{\set{S}_i} $ for all $j \in [2^m -1] \setminus \{1\}$ and $\bar{C}^{1} = \emptyset$. Then from (\ref{eq:enc_error}) and the union bound, we can write
\begin{align}
\Pr(E) & \le  \Pr(C^c_{\set{S}_1}\cap \bar{C}^1) + \cdots + \Pr(C^c_{\set{S}_{2^m-1}} \cap \bar{C}^{2^m-1})
\notag \\
\label{eq:prob_error_e}
& \le  \Pr( C^c_{\set{S}_1} | \bar{C}^1)  + \cdots + \Pr(C^c_{\set{S}_{2^m-1}} | \bar{C}^{2^m-1}).
\end{align}

Note that $\Pr(C^c_{\set{S}_j }|\bar{C}^{j})$, $j \in [2^m -1]$ represents the probability of the event that there is no  $U_{\set{S}_j}(\mb{k}_{\set{S}_j})$  $\epsilon_j$-jointly typical with $(U_{\set{S}_j}^{-}, X^n )$ given that for each $i<j$ we find $U_{\set{S}_i}(\mb{k}_{\set{S}_i})$
such that $U_{\set{S}_i}(\mb{k}_{\set{S}_i})$ is 
$\epsilon_i$-jointly typical with $(U_{\set{S}_i}^{-}, X^n )$, i.e.,
\begin{align*}
&\Pr(C^c_{\set{S}_j }|\bar{C}^{j}) = \Pr\left( \forall \mb{k}_{\set{S}_j},   U_{\set{S}_j}(\mb{k}_{\set{S}_j}) \notin \mathcal{T}^{(n)}_{\epsilon_j}(p| U_{\set{S}_j}^{-}, X^n )  | (U_{\set{S}_j}^{-}, X^n) \in \mathcal{T}^{(n)}_{\epsilon_{j-1}}(p) \right).
\end{align*}

 From Lemma \ref{thm:typical_joint} in Appendix \ref{app:typicality} and the inequality $(1 - \alpha)^\beta < e^{-\alpha\beta}$, we can write
\begin{align}
\Pr(C^c_{\set{S}_j }|\bar{C}^{j}) &< e^{-\left[(1 - \delta_{\epsilon_{j-1}, \epsilon_j}(n)) 2^{-n\left(I(X, U_{\set{S}_j}^{-} ; U_{\set{S}_j}) + 2\epsilon_jH(U_{\set{S}_j})\right)}2^{n(R_{\set{S}_j} + R'_{\set{S}_j}) } \right]}
\notag \\
\label{eq:enc_typ2}
& = e^{-\left[(1 - \delta_{\epsilon_{j-1}, \epsilon_j}(n) )2^{n\left((R_{\set{S}_j} + R'_{\set{S}_j}) - I(X, {U}_{\set{S}_j}^{-} ; U_{\set{S}_j}) - 2\epsilon_jH(U_{\set{S}_j})\right)}\right]}, 
\end{align}
 where $\delta_{\epsilon_{j-1}, \epsilon_j}(n)  \rightarrow 0 $ as $n \rightarrow \infty$. 
Note that when $H(U_{\set{S}_j}) = 0$,  $\Pr(C^c_{\set{S}_j }|\bar{C}^{j})$ is equal to zero. 
 Then $\Pr(C^c_{\set{S}_j }|\bar{C}^{j}) < \frac{\epsilon'}{2^m}$ if  $n \ge n_1(\epsilon', \epsilon_jH(U_{\set{S}_j})) $, and
\begin{align}
\label{ineq:cond_enc}
R_{\set{S}_j} + R'_{\set{S}_j} \ge  I(X, U_{\set{S}_j}^{-} ; U_{\set{S}_j}) + 3\epsilon_jH(U_{\set{S}_j}).
\end{align}

Hence, if $(R_{\set{S}_j}, R'_{\set{S}_j})$  satisfy the condition in (\ref{ineq:cond_enc}) for all $j \in [2^m-1]$, from (\ref{eq:prob_error_e}) we can conclude that the probability of error at the encoder satisfies
\begin{align}
\label{ineq:prob_enc_final}
\Pr(E) < \frac{2^m-1}{2^m}\epsilon'
\end{align}
when $n \ge N_1$ where $N_1 =  \max_{j \in [2^m-1]} n_1(\epsilon', \epsilon_jH(U_{\set{S}_j}))$.
\\
\textit{Error Analysis for Decoders :} Let us focus on decoder $l$ for some fixed $l \in [m]$. Decoding at this decoder is successful if the following conditions are satisfied:
\\
1. There is no error at the encoder.
\\
2. The source and the side information are  $\epsilon_0$-typical, i.e.,
\begin{align}
\label{eq:prob_cor_e0}
& D_0 =\left\{ (X^n,Y_1^n, \ldots, Y_m^n ) \in \mathcal{T}^{(n)}_{\epsilon_0}(p)\right\}.
\end{align}
3. The set of codewords $U_{\set{D}_l}(\mb{k}_{\set{D}_l}) = \{U_{\set{S}_j}(\mb{k}_{\set{S}_j}) | \set{S}_j \in \set{D}_l \}$ chosen by the encoder are $\epsilon_{l^*+1}$-letter typical with   $Y_l^n$, i.e.,
\begin{align}
\label{cond:dec1}
D_{1,l} = \left\{ \left(U_{\set{D}_l}(\mb{k}_{\set{D}_l}), X^n, Y^n_l \right) \in \mathcal{T}^{(n)}_{\epsilon_{l^*+1}}(p)\right\}.
\end{align}
4. Within the received bins $k_{\set{D}_l} =\{k_{\set{S}_j} | \set{S}_j \in \set{D}_l\}$, decoder $l$ can find a unique set of codewords, $  U_{\set{D}_l}(\mb{\wh{k}}_{\set{D}_l}) = \{U_{\set{S}_j}(\mb{\wh{k}}_{\set{S}_j}) | \wh{k}_{\set{S}_j} = k_{\set{S}_j}, \set{S}_j \in \set{D}_l \}$, such that 
$U_{\set{D}_l}(\mb{\wh{k}}_{\set{D}_l}) $ are $\epsilon_{l^*+1}$-letter typical with   $Y_l^n$, i.e.,
\begin{align}
\label{cond:dec2}
D_{2,l} = \left\{  \nexists \mb{\wt{k}}_{\set{D}_l} \neq \mb{k}_{\set{D}_l} \mbox{ such that }  \wt{k}_{\set{D}_l} = k_{\set{D}_l},  \left(U_{\set{D}_l}(\mb{\wt{k}}_{\set{D}_l}), Y^n_l \right) \in \mathcal{T}^{(n)}_{\epsilon_{l^*+1}}(p)\right\}.
\end{align}

Then we can write the probability of error at decoder $l$, denoted by $\Pr(D_{err,l})$, as
\begin{align}
\Pr(D_{err,l}) & = \Pr((E^c \cap D_0 \cap D_{1,l} \cap D_{2,l})^c)
\notag \\
& = \Pr(E \cup  D^c_0 \cup D^c_{1,l} \cup D^c_{2,l} )
\notag \\
& = \Pr(\bar{E} \cup (D^c_{1,l} \cap \bar{E}^c )\cup (D^c_{2,l} \cap \bar{E}^c \cap D_{1,l}) ), \mbox{ where } \bar{E} = E \cup  D^c_0 ,
\notag \\
\label{eq:prob_error_d}
& \le  \Pr(\bar{E}) + \Pr( D^c_{1,l}\cap \bar{E}^c) + \Pr( D^c_{2,l} \cap \bar{E}^c \cap  D_{1,l} ).
\end{align}

First we analyze $\Pr(\bar{E})$. By Lemma \ref{thm:typical1} in Appendix \ref{app:typicality}, $\Pr(D^c_0) < \delta_{\epsilon_0}(n)$ where $\delta_{\epsilon_0}(n) \rightarrow 0$ as $n \rightarrow \infty$. Then we can find $n_2(\epsilon', \delta_{\epsilon_0})$, $\epsilon' > 0$ such that if  $n \ge n_2(\epsilon', \delta_{\epsilon_0})$, $\Pr(D^c_0) < \frac{\epsilon'}{2^m}$. Hence, from (\ref{ineq:prob_enc_final}) and the union bound, we have
 \begin{align}
 \Pr(\bar{E}) \le \Pr(E) + \Pr(D^c_0) < \epsilon'
 \end{align}
when $n \ge \max\{ n_2(\epsilon', \delta_{\epsilon_0}),N_1\}$.

Now we focus on $\Pr( D^c_{1,l}\cap \bar{E}^c)$ and  $\Pr( D^c_{2,l} \cap 
\bar{E}^c \cap  D_{1,l} )$. 
We can upper bound $\Pr( D^c_{1,l}\cap \bar{E}^c)$ by
\begin{align}
\label{ineq:dec_markov}
\Pr\left(\left(U_{\set{D}_l}(\mb{k}_{\set{D}_l}), X^n ,Y^n_l \right) \notin \mathcal{T}^n_{\epsilon_{l^*+1}}(p) \big|
(U_{\set{D}_l}(\mb{k}_{\set{D}_l}), X^n ) \in \mathcal{T}^{(n)}_{\epsilon_{l^*}}(p) \right).
\end{align}
By Lemma \ref{thm:typical_markov} in Appendix \ref{app:typicality}, the probability  in (\ref{ineq:dec_markov}) is less than or equal to $\delta_{\epsilon_{l^*}\epsilon_{l^*+1}}(n)$ which goes to $0$ as $n \rightarrow \infty$.
Hence, $\Pr( D^c_{1,l}\cap \bar{E}^c) < \epsilon'$ if  $n \ge n_3(\epsilon',\delta_{\epsilon_{l^*}\epsilon_{l^*+1}} )$.

Now we consider $\Pr( D^c_{2,l} \cap \bar{E}^c \cap  D_{1,l} )$. Note that event $D^c_{2,l}$ can be rewritten as
\begin{align*}
&D^c_{2,l} = \bigcup_{\set{D}'_l :\set{D}'_l \subseteq  \set{D}_l, \set{D}'_l  \neq \emptyset} F_{\set{D}'_l}, \mbox{ where}
\\
&   F_{\set{D}'_l} = \bigg\{   \exists \mb{\wt{k}}_{\set{D}_l}  \mbox{ such that }\mb{\wt{k}}_{\set{S}_j} \neq \mb{k}_{\set{S}_j} \mbox{ for all } \set{S}_j \in \set{D}'_l ,\wt{k}_{\set{D}'_l} = k_{\set{D}'_l},\mb{\wt{k}}_{\set{S}_j} = \mb{k}_{\set{S}_j} \mbox{ for all } \set{S}_j \in  \set{D}_l \setminus\set{D}'_l \mbox{ and }
\\
& \quad \quad \quad \quad
   \left(U_{\set{D}_l}(\mb{\wt{k}}_{\set{D}_l}), Y^n_l \right) \in \mathcal{T}^{(n)}_{\epsilon_{l^*+1}}(p) \bigg\}.
\end{align*}

Using the union bound, we can write
\begin{align}
\label{ineq:prob_error_d2}
\Pr( D^c_{2,l} \cap \bar{E}^c \cap  D_{1,l} ) \le \sum_{\set{D}'_l :\set{D}'_l \subseteq  \set{D}_l, \set{D}'_l  \neq \emptyset} \Pr(F_{\set{D}'_l} \cap \bar{E}^c \cap  D_{1,l}).
\end{align}

Notice that $F_{\set{D}'_l} \cap \bar{E}^c \cap  D_{1,l}$ denotes the error event that there is no error at the encoder and the source and side information are $\epsilon_0$-typical (event $\bar{E}^c$), and decoder $l$ can find a set of indices $\{ \mb{\wh{k}}_{\set{S}_j} |  \set{S}_j \in \set{D}_l \}$ such that $U_{\set{D}_l}(\mb{\wh{k}}_{\set{D}_l})$ are $\epsilon_{l^* +1}$-jointly typical with $(X^n, Y^n_l)$ (event $D_{1,l}$); however the particular subset  $\mb{k}_{\set{D}'_l} = \{\mb{k}_{\set{S}_j} |  \set{S}_j \in \set{D}'_l \}$ of those indices is not unique (event $F_{\set{D}'_l}$). Now we bound each term inside the summation in (\ref{ineq:prob_error_d2}). To do this, first we define an event $\bar{F}_{\set{D}'_l}$ by replacing the typical set $\mathcal{T}^{(n)}_{\epsilon_{l^*+1}}(p) $ in event $F_{\set{D}'_l}$ with $\mathcal{T}^{(n)}_{\epsilon_{l^*+2}}(p)$. In other words,
\begin{align*}
&\bar{F}_{\set{D}'_l}  =  \bigg\{   \exists \mb{\wt{k}}_{\set{D}_l}  \mbox{ such that }\mb{\wt{k}}_{\set{S}_j} \neq \mb{k}_{\set{S}_j} \mbox{ for all } \set{S}_j \in \set{D}'_l ,\wt{k}_{\set{D}'_l} = k_{\set{D}'_l},\mb{\wt{k}}_{\set{S}_j} = \mb{k}_{\set{S}_j} \mbox{ for all } \set{S}_j \in  \set{D}_l \setminus\set{D}'_l \mbox{ and }
\\
& \quad \quad \quad \quad
   \left(U_{\set{D}_l}(\mb{\wt{k}}_{\set{D}_l}), Y^n_l \right) \in \mathcal{T}^{(n)}_{\epsilon_{l^*+2}}(p) \bigg\},
\end{align*}
giving $F_{\set{D}'_l} \subseteq \bar{F}_{\set{D}'_l} $. 
Let $S_1 =\{
 \mb{\wt{k}}_{\set{D}'_l}| \wt{k}'_{\set{S}_j} \neq k'_{\set{S}_j}, 
  \wt{k}_{\set{D}'_l} = k_{\set{D}'_l},
 \forall \set{S}_j \in \set{D}'_l
  \}$ and $S_2 =\{
 \mb{\wt{k}}_{\set{D}'_l}| 
   \wt{k}_{\set{S}_j} = 1,\forall \set{S}_j \in \set{D}'_l \}$. Then we can write
\begin{align}
\label{ineq:err_prob_dec2}
 \Pr(\bar{F}_{\set{D}'_l} \cap \bar{E}^c \cap  D_{1,l}) &\le
  \Pr\left( \bigcup_{ S_1 } 
 U_{\set{D}'_l}(\mb{\wt{k}}_{\set{D}'_l}) \in  \mathcal{T}^{(n)}_{\epsilon_{l^*+2}}(p| U_{\set{D}_l\setminus \set{D}'_l}(\mb{k}_{\set{D}_l\setminus \set{D}'_l}),Y^n_l) \big| ( U_{\set{D}_l\setminus \set{D}'_l}(\mb{k}_{\set{D}_l\setminus \set{D}'_l}),Y^n_l) \in \mathcal{T}^{(n)}_{\epsilon_{l^*+1}}(p)\right)
\\
  \label{ineq:err_prob_dec}
 & \le 
  \Pr\left( \bigcup_{ S_2 } 
 U_{\set{D}'_l}(\mb{\wt{k}}_{\set{D}'_l}) \in  \mathcal{T}^{(n)}_{\epsilon_{l^*+2}}(p| U_{\set{D}_l\setminus \set{D}'_l}(\mb{k}_{\set{D}_l\setminus \set{D}'_l}),Y^n_l) \big| ( U_{\set{D}_l\setminus \set{D}'_l}(\mb{k}_{\set{D}_l\setminus \set{D}'_l}),Y^n_l) \in \mathcal{T}^{(n)}_{\epsilon_{l^*+1}}(p)\right),
 \end{align}
 where (\ref{ineq:err_prob_dec}) is obtained by 
using Lemma \ref{lemma:bin_indp} in Appendix \ref{app:typicality}. Then due to the union bound of probabilities we can write  
 \begin{align}
  \Pr(\bar{F}_{\set{D}'_l} \cap \bar{E}^c \cap  D^c_{1,l}) &
  \le \sum_{ S_2}  \Pr\left(  U_{\set{D}'_l}(\mb{\wt{k}}_{\set{D}'_l}) \in  \mathcal{T}^{(n)}_{\epsilon_{l^*+2}}(p| U_{\set{D}_l\setminus \set{D}'_l}(\mb{k}_{\set{D}_l\setminus \set{D}'_l}),Y^n_l) \big| ( U_{\set{D}_l\setminus \set{D}'_l}(\mb{k}_{\set{D}_l\setminus \set{D}'_l}),Y^n_l) \in \mathcal{T}^{(n)}_{\epsilon_{l^*+1}}(p)\right)
\notag  \\
 & \le 
 \label{ineq:ach_proof_dec_err}
 2^{n\sum_{
  \substack{ \set{S}_j \in \set{D}'_l}} R'_{\set{S}_j}
  }
  2^{-n\left(\sum_{\set{S}_j \in \set{D}'_l} H(U_{\set{S}_j}) - H(U_{\set{D}'_l} | U_{\set{D}_l \setminus \set{D}'_l}, Y_l)-2 \epsilon_{l^* + 2} \left(\sum_{\set{S}_j \in \set{D}'_l} H(U_{\set{S}_j}) \right) \right)},
\mbox{  from Corollary \ref{corr:typical_joint}}.
\end{align}
Note that $R'_{\set{S}_j} \ge 0, \mbox{ for all } j \in [2^m-1]$ and when each $ R'_{\set{S}_j} = 0$, $\set{S}_j \in \set{D}'_l$, there is only one codeword $U(\mb{k}_{\set{S}_j})$, $\set{S}_j \in \set{D}'_l$ in each bin.  Then,  from (\ref{ineq:err_prob_dec2})
 $\Pr(\bar{F}_{\set{D}'_l} \cap \bar{E}^c \cap  D_{1,l}) = 0$ in this case. Also, when each $H(U_{\set{S}_j}) = 0$, $\set{S}_j \in \set{D}'_l$,  $\Pr(\bar{F}_{\set{D}'_l} \cap \bar{E}^c \cap  D_{1,l})$ is equal to $0$.

Thus from (\ref{ineq:ach_proof_dec_err}), if 
\begin{align}
\label{ineq:cond_dec2}
 \sum_{\set{S}_j \in \set{D}'_l} R'_{\set{S}_j} &\le
\max \left\{ \left( \sum_{\set{S}_j \in \set{D}'_l} H(U_{\set{S}_j}) \right)- H(U_{\set{D}'_l} | U_{\set{D}_l \setminus \set{D}'_l}, Y_l)-3 \epsilon_{l^* + 2} \left(\sum_{\set{S}_j \in \set{D}'_l} H(U_{\set{S}_j}) \right), 0
 \right\}
\end{align} 
and $n \ge n_4(\epsilon',\epsilon_{l^* + 2}, H(U_{\set{S}_j}) )$, $\Pr( D^c_{2,l} \cap \bar{E}^c \cap  D^c_{1,l} ) < \frac{\epsilon'}{2^{|\set{D}_l|}}$. Then from (\ref{eq:prob_error_d}), if $R'_{\set{S}_j}$ satisfies (\ref{ineq:cond_dec2}) for all $\set{D}_l$, $l \in [m]$ and $n > N$, where $ N = \max\{N_1, n_2(\epsilon', \delta_{\epsilon_0}),n_3(\epsilon',\delta_{\epsilon_{l^*}\epsilon_{l^*+1}} ),\max_{l \in [m]}\{n_4(\epsilon',\epsilon_{l^* + 2},H(U_{\set{S}_j}) )\}\}$
\begin{align}
\label{ineq:Derr_l}
\Pr(D_{err,l}) < 3\epsilon' .
\end{align}
Let 
\begin{align*}
D_{err} = \cup_{l \in [m]} D_{err,l}
\end{align*}
denote the event that there is a decoding error at some decoder. By (\ref{ineq:Derr_l}) and the union bound we have
\begin{align}
\label{ineq:Derr_l_union}
\Pr(D_{err}) < 3\epsilon'm .
\end{align}
Thus there must exist a single code in the ensemble for which (\ref{ineq:Derr_l_union}) holds.
Now we focus on the distortion constraints at the decoder for this particular code. 
Assuming that there is no error occurring at the encoder and the decoders (corresponding to event $E^c \cap D^c_{err,l}$), decoder $l$ can find a unique  $u_{\set{D}_l}(\mb{k}_{\set{D}_l})$ such that $(u_{\set{D}_l}(\mb{k}_{\set{D}_l}), y^n_l, x^n )$ is $\epsilon_{l^* +1}$-jointly typical and it can reconstruct $\hat{x}^n_l$ symbol-by-symbol through $\hat{x}_{li}= g_l(u_{\set{D}_li},y_{li})$, $i \in [n]$.  Then using the arguments in 
\cite[page 57]{kramer} we can bound the average distortion at decoder $l$ by 
\begin{align}
 \frac{1}{n} 
\sum_{i = 1}^{n} d_l(x_i,\wh{x}_{li})
& = \sum_{i = 1}^{n} d_l(x_i, g_l(u_{\set{D}_li},y_{li}))
\notag \\
& \le E\left[ d_l(X, g_l(U_{\set{D}_l}, Y_{l}))\right] + \epsilon_{l^* +1}D_{l,max}
\notag \\
& \le D_l + \epsilon_{l^* +1}D_{l,max},
\end{align}
where $D_{l,max}$ is the maximum distortion that $d_l(\cdot,\cdot)$ can give. 
Then the expected distortion at decoder $l$ can be bounded by
\begin{align}
\label{ineq:avg_dist_const}
E\left[ \frac{1}{n} \sum_{i = 1}^{n} d_l(x_i,\wh{x}_{li})\right] 
& \le (D_l + \epsilon_{l^* +1}D_{l,max})\Pr(E^c \cap D^c_{err,l}) + D_{l,max}\Pr(E \cup D_{err,l})
\notag \\
& \le D_l + D_{l,max}(\epsilon_{l^* +1}+ \Pr(E \cup D_{err,l}))
\notag \\
& < D_l +  D_{l,max}(\epsilon_{l^* +1} + 4\epsilon'm),
\end{align}
where (\ref{ineq:avg_dist_const}) holds if $n> N$ and $(R_{\set{S}_j}, R'_{\set{S}_j})$, $\set{S}_j \subseteq [m]$ satisfy the conditions in (\ref{ineq:cond_enc}), (\ref{ineq:cond_dec2}), and the following non-negativity conditions :
 \begin{align}
 \label{cond:nonneg1}
& R_{\set{S}_j} \ge 0, \mbox{ for all } j \in [2^m-1] 
\\
 \label{cond:nonneg2}
 & R'_{\set{S}_j} \ge 0, \mbox{ for all } j \in [2^m-1].
 \end{align}


Thus for all sufficiently large $n$, there exists a code whose expected distortion
at decoder $l$ satisfies (\ref{ineq:avg_dist_const}) and whose rate does not exceed
\begin{align}
\label{lowerLP:1}
&\inf \sum^{2^m-1}_{j =1} R_{\set{S}_j}
\\
& \mbox{ subject to :}  R_{\set{S}_j},  R'_{\set{S}_j}, j \in [2^m-1] \mbox{ satisfying }  (\ref{ineq:cond_enc}), (\ref{ineq:cond_dec2}), (\ref{cond:nonneg1}), \mbox{ and }  (\ref{cond:nonneg2}).
\notag
\end{align}

\begin{lemma}
\label{lemma:LP_inequality}
Let $0<\epsilon_0<\epsilon_1< \ldots <\epsilon_{2^m +1}$, and  $U_{\set{S}_j}$,$\set{S}_j \in v$, be as in the proof of Theorem \ref{thm:gen_ach}. For $\gamma \ge 0$, consider the following linear program:
\begin{align}
\label{lowerLP:2}
&\wt{R}(\gamma) = \inf_{C^{LP}_{ach}(\gamma)} \sum^{2^m-1}_{j =1} R_{\set{S}_j}, 
\end{align}
where $C^{LP}_{ach}(\gamma)$ denotes the set of $R_{\set{S}_j}$ and $R'_{\set{S}_j}$ such that
\\
1) $ R_{\set{S}_j} \ge 0$ and $R'_{\set{S}_j} \ge 0$, for all $j \in [2^m-1]$;
\\
2) $R_{\set{S}_j} + R'_{\set{S}_j} \ge  I(X, U_{\set{S}_j}^{-} ; U_{\set{S}_j}) + 3\gamma$,  for all $j \in [2^m-1]$;
\\
3) For each decoder $l$, $l \in [m]$
\begin{align*}
 \sum_{\set{S}_j \in \set{D}'_l} R'_{\set{S}_j} &\le
\max \left\{ \left( \sum_{\set{S}_j \in \set{D}'_l} H(U_{\set{S}_j}) \right)- H(U_{\set{D}'_l} | U_{\set{D}_l \setminus \set{D}'_l}, Y_l)-3(2^m-1)\gamma, 0
 \right\}.
\end{align*}

Then $\wt{R}(\gamma)$ is continuous at $\gamma = 0$ and is greater than or equal to the optimal value in (\ref{lowerLP:1}) if 
\begin{align}
\label{third_jour:gamma_cont_cond}
\gamma \ge \epsilon_{2^m +1}\max_{U_{\set{S}_j}}H(U_{\set{S}_j}).
\end{align}
\end{lemma}
\begin{proof}[Proof of Lemma \ref{lemma:LP_inequality}]
Note that when $\gamma = 0$, $C^{LP}_{ach}(\gamma)$ is equal to $C^{LP}_{ach}$. Also, since the alphabets are finite, $C^{LP}_{ach}(\gamma)$ is nonempty for any $\gamma \ge 0$. The continuity of $\wt{R}(\gamma)$ in $\gamma$ then follows from standard results on the continuity of LPs \cite{cont_LP}. The relation with (\ref{lowerLP:1}) follows by noting that $C^{LP}_{ach}(\gamma)$ is contained in the set defined by the constraints  (\ref{ineq:cond_enc}), (\ref{ineq:cond_dec2}), (\ref{cond:nonneg1}), \mbox{ and }  (\ref{cond:nonneg2}), whenever (\ref{third_jour:gamma_cont_cond}) holds. 
\end{proof}

Now given $\epsilon > 0$, choose  $0<\epsilon_0<\epsilon_1< \ldots <\epsilon_{2^m +1}$, $\epsilon'$ and $\gamma$ such that 
\begin{align*}
& D_{l,max}(\epsilon_{l^*+1} + 4\epsilon'm) < \epsilon \mbox{ for all } l \in [m]
\\
&\gamma \ge \epsilon_{2^m +1}\max_{U_{\set{S}_j}}H(U_{\set{S}_j})
\end{align*}
and $\wt{R}(\gamma) < \wt{R}(0) + \epsilon$. Then we have that for all sufficiently large $n$, there exists a code with rate at most  $ \wt{R}(0) + \epsilon$ whose expected distortion at decoder $l$ is at most $D_l + \epsilon$. It follows that 
$\wt{R}(0)$ is $\mb{D}$-achievable as desired and hence $R'_{ach}(\mb{D})$ is $\mb{D}$-achievable. Lastly, since $R(\mb{D})$ is a convex function with respect to $\mb{D}$ and it is upper bounded by $R'_{ach}(\mb{D})$, $R(\mb{D})$ must lie beneath the lower convex envelope of $R'_{ach}(\mb{D})$.
\end{proof}
\section{}
\label{app:strict_ineq}
\begin{proof}[Proof of Lemma \ref{lemma:minimax_odd_cycle_gauss}]
Let $\epsilon > 0$ be given. Note that since each side information variable is a function of the source, the set $\bar{P}$ in Theorem \ref{theorem:minmax} contains only one element. Then  let us select $V = \emptyset$ and $U_{Y_i} = X_i + N_i$ for all $i \in [m]$ where $N_i$ is independent of $\mb{X}$ and the rest of the $N_j$'s, $j \neq i$ and is such that $K_{X_i |U_{Y_i}} = D + \epsilon$. Then the $U_{Y_i}$'s are feasible in the optimization in Theorem \ref{theorem:minmax} and we can write
\begin{align}
R^m_{lb}(\mb{D}+ \epsilon\mb{1}) \le \max_{\sigma} R_{\sigma}
\end{align}
where 
\begin{align}
 &R_{\sigma} = I(\mbt{X}; {U_{Y_\sigma(1)}}|\mbt{Y}_{\sigma(1)})
 + I( \mbt{X} ; {U_{Y_\sigma(2)}}|{U_{Y_\sigma(1)}},\mbt{Y}_{\sigma(1)},\mbt{Y}_{\sigma(2)}) +\cdots
\notag \\
&\phantom{=} \quad + I( \mbt{X} ; {U_{Y_\sigma(m)}}|  {U_{Y_\sigma(1)}},\ldots, U_{Y_\sigma(m-1)},\mbt{Y}_{\sigma(1)}, \ldots, \mbt{Y}_{\sigma(m)}),
\end{align}
and $\sigma(.)$ denotes a permutation on the set $[m]$.
Using the chain rule and since $U_{Y_i} = X_i + N_i$, we can write $R_{\sigma}$  as
\begin{align}
\label{app:lemma:strict_ineq:ineq1}
& I(X_{\sigma(1)}; U_{Y_\sigma(1)}|\mbt{Y}_{\sigma(1)})
 + I( X_{\sigma(2)} ; U_{Y_\sigma(2)}|  \mbt{Y}_{\sigma(1)}, \mbt{Y}_{\sigma(2)}) +\cdots
\notag \\
&\quad + I( X_{\sigma(m)}; U_{Y_\sigma(m)}|  \mbt{Y}_{\sigma(1)}, \ldots, \mbt{Y}_{\sigma(m)}),
\end{align}
where each mutual information term is equal to either $\frac{1}{2}\log\frac{1}{D+ \epsilon}$ or $0$.

Now we show that $\max_{\sigma} R_{\sigma}$, where  $R_{\sigma}$ is equal to (\ref{app:lemma:strict_ineq:ineq1}),   equals $R^m_{lb}(\mb{D}+ \epsilon\mb{1})$. Note that $\bar{R}_{lb}$ in Theorem \ref{theorem:minmax} can be written as
\begin{align*}
\max_{\sigma} \bar{R}_{\sigma},
\end{align*}
where $\bar{R}_{\sigma}$ is equal to the right-hand side of (\ref{4lower_general}) for a given $\sigma$.
Using a series of chain rules and expanding the mutual information terms, we can rewrite $\bar{R}_{\sigma}$ as \footnote{We interpret the differential
entropy of an empty set of continuous random variables to be zero.}
\begin{align}
&  I(\mbt{X} \setminus \mbt{Y}_{\sigma(1)}, \mbt{Y}_{\sigma(1)}; V,{U_{Y_\sigma(1)}}|\mbt{Y}_{\sigma(1)})
 + I( \mbt{X} \setminus \cup^2_{i=1}\mbt{Y}_{\sigma(i)},  \cup^2_{i=1}\mbt{Y}_{\sigma(i)} ; {U_{Y_\sigma(2)}}|V,{U_{Y_\sigma(1)}},\mbt{Y}_{\sigma(1)},\mbt{Y}_{\sigma(2)}) +\cdots
\notag \\
&\quad + I( \mbt{X}\setminus  \cup^m_{i=1}\mbt{Y}_{\sigma(i)},  \cup^m_{i=1}\mbt{Y}_{\sigma(i)} ; {U_{Y_\sigma(m)}}| V, {U_{Y_\sigma(1)}},\ldots, U_{Y_\sigma(m-1)},\mbt{Y}_{\sigma(1)}, \ldots, \mbt{Y}_{\sigma(m)})
\\
&= I(\mbt{X} \setminus \mbt{Y}_{\sigma(1)}; V,{U_{Y_\sigma(1)}}|\mbt{Y}_{\sigma(1)})
 + I( \mbt{X} \setminus \cup^2_{i=1}\mbt{Y}_{\sigma(i)}; {U_{Y_\sigma(2)}}|V,{U_{Y_\sigma(1)}},\mbt{Y}_{\sigma(1)},\mbt{Y}_{\sigma(2)}) +\cdots
\notag \\
&\quad + I( \mbt{X}\setminus  \cup^m_{i=1}\mbt{Y}_{\sigma(i)} ; {U_{Y_\sigma(m)}}| V, {U_{Y_\sigma(1)}},\ldots, U_{Y_\sigma(m-1)},\mbt{Y}_{\sigma(1)}, \ldots, \mbt{Y}_{\sigma(m)})
\\
&\overset{a}{=} I(\mbt{X} \setminus \mbt{Y}_{\sigma(1)}; V,{U_{Y_\sigma(1)}}|\mbt{Y}_{\sigma(1)})
 + I( \mbt{X} \setminus \cup^2_{i=1}\mbt{Y}_{\sigma(i)}; {U_{Y_\sigma(2)}}|V, {U_{Y_\sigma(1)}},\mbt{Y}_{\sigma(1)},\mbt{Y}_{\sigma(2)}) +\cdots
\notag \\
&\quad + I( \mbt{X}\setminus  \cup^k_{i=1}\mbt{Y}_{\sigma(i)} ; {U_{Y_\sigma(k)}}| V,  {U_{Y_\sigma(1)}},\ldots, U_{Y_\sigma(k-1)},\mbt{Y}_{\sigma(1)}, \ldots, \mbt{Y}_{\sigma(k)})
\\
&\overset{b}{=} h(\mbt{X} \setminus \mbt{Y}_{\sigma(1)}|\mbt{Y}_{\sigma(1)}) - h(\mbt{X} \setminus \mbt{Y}_{\sigma(1)}| V, {U_{Y_\sigma(1)}},\mbt{Y}_{\sigma(1)})
\notag \\
&\quad  + h( \mbt{X} \setminus \cup^2_{i=1}\mbt{Y}_{\sigma(i)}|V, {U_{Y_\sigma(1)}},\mbt{Y}_{\sigma(1)},\mbt{Y}_{\sigma(2)}) -  h( \mbt{X} \setminus \cup^2_{i=1}\mbt{Y}_{\sigma(i)}| V, {U_{Y_\sigma(1)}},{U_{Y_\sigma(2)}},\mbt{Y}_{\sigma(1)},\mbt{Y_{\sigma(2)}}) +\cdots
\notag \\
&\quad + h( \mbt{X}\setminus  \cup^k_{i=1}\mbt{Y}_{\sigma(i)} | V, {U_{Y_\sigma(1)}},\ldots, U_{Y_\sigma(k-1)},\mbt{Y}_{\sigma(1)}, \ldots, \mbt{Y}_{\sigma(k)}) 
\notag
\\
& \quad - h( \mbt{X}\setminus  \cup^k_{i=1}\mbt{Y}_{\sigma(i)} |  V, {U_{Y_\sigma(1)}},\ldots, U_{Y_\sigma(k)},\mbt{Y}_{\sigma(1)}, \ldots, \mbt{Y}_{\sigma(k)})
\\
&\overset{c}{=} h(\mbt{X} \setminus \mbt{Y}_{\sigma(1)}|\mbt{Y}_{\sigma(1)}) 
- h(\mbt{Y}_{\sigma(2)} \setminus \mbt{Y}_{\sigma(1)}| V, {U_{Y_\sigma(1)}},\mbt{Y}_{\sigma(1)})
- h( \mbt{X} \setminus \cup^2_{i=1}\mbt{Y}_{\sigma(i)}| V, {U_{Y_\sigma(1)}},\mbt{Y}_{\sigma(1)},\mbt{Y_{\sigma(2)}})
\notag \\
&\quad  + h( \mbt{X} \setminus \cup^2_{i=1}\mbt{Y}_{\sigma(i)}|{V, U_{Y_\sigma(1)}},\mbt{Y}_{\sigma(1)},\mbt{Y}_{\sigma(2)})
\notag 
\\
 &\quad
 -  h(\mbt{Y}_{\sigma(3)}\setminus \cup^2_{i=1}\mbt{Y}_{\sigma(i)}| V, {U_{Y_\sigma(1)}},{U_{Y_\sigma(2)}},\mbt{Y}_{\sigma(1)},\mbt{Y}_{\sigma(2)}) 
  -  h( \mbt{X} \setminus \cup^3_{i=1}\mbt{Y}_{\sigma(i)}| V, {U_{Y_\sigma(1)}},{U_{Y_\sigma(2)}},\mbt{Y}_{\sigma(1)},\mbt{Y}_{\sigma(2)},\mbt{Y}_{\sigma(3)}) +
\notag \\
&\quad \cdots + h( \mbt{X}\setminus  \cup^k_{i=1}\mbt{Y}_{\sigma(i)} | V, {U_{Y_\sigma(1)}},\ldots, U_{Y_\sigma(k-1)},\mbt{Y}_{\sigma(1)}, \ldots, \mbt{Y}_{\sigma(k)}) 
\notag
\\
& \quad - h( \mbt{X}\setminus  \cup^k_{i=1}\mbt{Y}_{\sigma(i)} | V, {U_{Y_\sigma(1)}},\ldots, U_{Y_\sigma(k)},\mbt{Y}_{\sigma(1)}, \ldots, \mbt{Y}_{\sigma(k)})
\\
\label{app:lemma:strict_ineq:ineq2}
& =  h(\mbt{X} \setminus  \mbt{Y}_{\sigma(1)} | \mbt{Y}_{\sigma(1)}) - h(\mbt{Y}_{\sigma(2)} \setminus \mbt{Y}_{\sigma(1)}| V,U_{Y_\sigma(1)},\mbt{Y}_{\sigma(1)})
-  h(\mbt{Y}_{\sigma(3)} \setminus  \cup^2_{i=1}\mbt{Y}_{\sigma(i)} | V,{U_{Y_\sigma(1)}},U_{Y_\sigma(2)}, \mbt{Y}_{\sigma(1)}, \mbt{Y}_{\sigma(2)})-\cdots
\notag \\
&\quad - h( \mbt{Y}_{\sigma(k)} \setminus  \cup^{k-1}_{i=1}\mbt{Y}_{\sigma(i)} | V, {U_{Y_\sigma(1)}},\ldots, U_{Y_\sigma(k-1)}, \mbt{Y}_{\sigma(1)}, \ldots, \mbt{Y}_{\sigma(k-1)}) 
\notag \\
&\quad 
-h(\mbt{X} \setminus  \cup^k_{i=1}\mbt{Y}_{\sigma(i)} | V,{U_{Y_\sigma(1)}},\ldots, U_{Y_\sigma(k)},  \mbt{Y}_{\sigma(1)}, \ldots, \mbt{Y}_{\sigma(k)})
\\
\label{app:lemma:strict_ineq:ineq3}
& \ge  h(\mbt{X} \setminus \mbt{Y}_{\sigma(1)} |\mbt{Y}_{\sigma(1)}) - h(\mbt{Y}_{\sigma(2)} \setminus \mbt{Y}_{\sigma(1)} | \wh{X}_{\sigma(1)},\mbt{Y}_{\sigma(1)})
-  h(\mbt{Y}_{\sigma(3)} \setminus  \cup^2_{i=1}\mbt{Y}_{\sigma(i)}| \wh{X}_{\sigma(1)},\wh{X}_{\sigma(2)}, \mbt{Y}_{\sigma(1)}, \mbt{Y}_{\sigma(2)})+\cdots
\notag \\
&\quad - h( \mbt{Y}_{\sigma(k)} \setminus  \cup^{k-1}_{i=1}\mbt{Y}_{\sigma(i)} |  \wh{X}_{\sigma(1)},\ldots,  \wh{X}_{\sigma(k-1)}, \mbt{Y}_{\sigma(1)}, \ldots, \mbt{Y}_{\sigma(k-1)}) 
\notag \\
&\quad 
-h(\mbt{X} \setminus  \cup^k_{i=1}\mbt{Y}_{\sigma(i)} |  \wh{X}_{\sigma(1)},\ldots,  \wh{X}_{\sigma(k)},  \mbt{Y}_{\sigma(1)}, \ldots, \mbt{Y}_{\sigma(k)}), \\
\notag
 & = h(\mbt{X} \setminus \mbt{Y}_{\sigma(1)}|\mbt{Y}_{\sigma(1)})
    - \sum_{j = 1}^{k-1}
      h(\mbt{Y}_{\sigma(j+1)} \setminus \cup_{i=1}^{j}\mbt{Y}_{\sigma(i)}|
        \wh{X}_{\sigma(1)},\ldots, \wh{X}_{\sigma(j)}, 
              \mbt{Y}_{\sigma(1)}, \ldots, \mbt{Y}_{\sigma(j)}) \\
  & \quad   - h(\mbt{X} \setminus \cup_{i=1}^{k}\mbt{Y}_{\sigma(i)}|
        \wh{X}_{\sigma(1)},\ldots, \wh{X}_{\sigma(k)}, 
              \mbt{Y}_{\sigma(1)}, \ldots, \mbt{Y}_{\sigma(k)}),
\label{app:lemma:summation}
\end{align}
where $\wh{X}_{\sigma(i)}$ is such that $K_{{X}_{\sigma(i)}| \wh{X}_{\sigma(i)},
\mbt{Y_{\sigma(i)}} } \le D + \epsilon$, $N_{\sigma(i)}$ is as defined before, and
\\
a : follows since $\mbt{Y}_{i} = ({X_{i-1}}, {X_{i+1}})$ and there exists a $k > 1$ such that $\mbt{X}\setminus \cup^l_{i=1}\mbt{Y}_{\sigma(i)} = \emptyset$ for all $l > k$.
\\
b : follows by expanding each mutual information term.
\\
c : follows by applying the chain rule to all minus terms except the last one (i.e., the second term, the fourth term, etc.).

Now 
\begin{align}
\notag
& h(\mbt{Y}_{\sigma(j+1)} \setminus \cup_{i=1}^{j}\mbt{Y}_{\sigma(i)}|
   \wh{X}_{\sigma(1)},\ldots, \wh{X}_{\sigma(j)}, 
         \mbt{Y}_{\sigma(1)}, \ldots, \mbt{Y}_{\sigma(j)}) \\
& \le 
\label{app:lemma:otherX}
h(\mbt{X}_{\sigma(j+1)-1} \setminus \cup_{i=1}^{j}\mbt{Y}_{\sigma(i)}|
   \wh{X}_{\sigma(1)},\ldots, \wh{X}_{\sigma(j)}, 
         \mbt{Y}_{\sigma(1)}, \ldots, \mbt{Y}_{\sigma(j)}) \\
\notag
& \quad + 
h(\mbt{X}_{\sigma(j+1)+1} \setminus \cup_{i=1}^{j}\mbt{Y}_{\sigma(i)}|
   \wh{X}_{\sigma(1)},\ldots, \wh{X}_{\sigma(j)}, 
         \mbt{Y}_{\sigma(1)}, \ldots, \mbt{Y}_{\sigma(j)}).
\end{align}
We shall show that
\begin{align}
\notag
& h(\mbt{X}_{\sigma(j+1)-1} \setminus \cup_{i=1}^{j}\mbt{Y}_{\sigma(i)}|
   \wh{X}_{\sigma(1)},\ldots, \wh{X}_{\sigma(j)}, 
         \mbt{Y}_{\sigma(1)}, \ldots, \mbt{Y}_{\sigma(j)}) \\
 &  \le
h(\mbt{X}_{\sigma(j+1)-1} \setminus \cup_{i=1}^{j}\mbt{Y}_{\sigma(i)}|
   {X}_{\sigma(1)} + {N}_{\sigma(1)},\ldots, 
   {X}_{\sigma(j)} + {N}_{\sigma(j)}, 
    \mbt{Y}_{\sigma(1)}, \ldots, \mbt{Y}_{\sigma(j)})
\label{app:lemma:target}
\end{align}
for all $j \in \{1,\ldots,k-1\}$ and similarly for the quantities
\begin{align}
\label{app:lemma:other1}
h(\mbt{X}_{\sigma(j+1)+1} \setminus \cup_{i = 1}^j \mbt{Y}_{\sigma(i)}|
    \wh{X}_{\sigma(1)}, \ldots, \wh{X}_{\sigma(j)},
     \mbt{Y}_{\sigma(1)}, \ldots, \mbt{Y}_{\sigma(j)}).
\intertext{and}
\label{app:lemma:other2}
h(\mbt{X} \setminus \cup_{i = 1}^k \mbt{Y}_{\sigma(i)}|
    \wh{X}_{\sigma(1)}, \ldots, \wh{X}_{\sigma(k)},
     \mbt{Y}_{\sigma(1)}, \ldots, \mbt{Y}_{\sigma(k)}),
\end{align}
appearing in (\ref{app:lemma:otherX}) and (\ref{app:lemma:summation}), respectively.
To show (\ref{app:lemma:target}), fix $j \in \{1,\ldots,k-1\}$ and define the
sets of indices \footnote{Recall that here $x \mod m$ is defined
to be in $[m]$.}
\begin{align}
I_1 & = \cup_{i = 1}^j \{\sigma(i) - 1 \mod m, \sigma(i) +1 \mod m\} \\
I_2 & = \cup_{i = 1}^j \{\sigma(i)\}.
\end{align}
If $\sigma(j+1) - 1 \mod m\in I_1$, then the entropy quantities 
on both sides
of (\ref{app:lemma:target}) are empty so (\ref{app:lemma:target}) 
trivially holds.
If $\sigma(j+1) - 1 \mod m\in I_2 \setminus I_1$, then we have
\begin{align}
& h(\mbt{X}_{\sigma(j+1)-1} \setminus \cup_{i=1}^{j}\mbt{Y}_{\sigma(i)}|
   \wh{X}_{\sigma(1)},\ldots, \wh{X}_{\sigma(j)}, 
         \mbt{Y}_{\sigma(1)}, \ldots, \mbt{Y}_{\sigma(j)}) \\
  & \le \frac{1}{2} \log (2 \pi e (D + \epsilon)) \\
  & = h(\mbt{X}_{\sigma(j+1)-1} \setminus \cup_{i=1}^{j}\mbt{Y}_{\sigma(i)}|
   {X}_{\sigma(1)} + {N}_{\sigma(1)},\ldots, 
   {X}_{\sigma(j)} + {N}_{\sigma(j)}, 
    \mbt{Y}_{\sigma(1)}, \ldots, \mbt{Y}_{\sigma(j)}).
\end{align}
And if $\sigma(j + 1) - 1 \mod m \notin I_1 \cup I_2$, then we have
\begin{align}
& h(\mbt{X}_{\sigma(j+1)-1} \setminus \cup_{i=1}^{j}\mbt{Y}_{\sigma(i)}|
   \wh{X}_{\sigma(1)},\ldots, \wh{X}_{\sigma(j)}, 
         \mbt{Y}_{\sigma(1)}, \ldots, \mbt{Y}_{\sigma(j)}) \\
   & \le h(\mbt{X}_{\sigma(j+1)-1} \setminus \cup_{i=1}^{j}\mbt{Y}_{\sigma(i)})
      \\
  & = h(\mbt{X}_{\sigma(j+1)-1} \setminus \cup_{i=1}^{j}\mbt{Y}_{\sigma(i)}|
   {X}_{\sigma(1)} + {N}_{\sigma(1)},\ldots, 
   {X}_{\sigma(j)} + {N}_{\sigma(j)}, 
    \mbt{Y}_{\sigma(1)}, \ldots, \mbt{Y}_{\sigma(j)}).
\end{align}
This establishes (\ref{app:lemma:target}). The argument for the quantities
in (\ref{app:lemma:other1}) and (\ref{app:lemma:other2}) 
is similar. Substituting into
(\ref{app:lemma:summation}) gives
\begin{align}
\label{app:lemma:bottomline}
\bar{R}_\sigma
& \ge  h(\mbt{X} \setminus  \mbt{Y}_{\sigma(1)} |\mbt{Y}_{\sigma(1)}) - h(\mbt{Y}_{\sigma(2)} \setminus \mbt{Y}_{\sigma(1)} | X_{\sigma(1)} +N_{\sigma(1)}, \mbt{Y}_{\sigma(1)})
\notag \\
& \quad
-  h(\mbt{Y}_{\sigma(3)} \setminus  \cup^2_{i=1}\mbt{Y}_{\sigma(i)} | X_{\sigma(1)} +N_{\sigma(1)}, X_{\sigma(2)} +N_{\sigma(2)} ,\mbt{Y}_{\sigma(1)}, \mbt{Y}_{\sigma(2)}) - \cdots
\notag \\
&\quad 
- h( \mbt{Y}_{\sigma(k)} \setminus  \cup^{k-1}_{i=1}\mbt{Y}_{\sigma(i)} |  X_{\sigma(1)} +N_{\sigma(1)}, \ldots,  X_{\sigma(k-1)} +N_{\sigma(k-1)}, \mbt{Y}_{\sigma(1)}, \ldots, \mbt{Y}_{\sigma(k-1)}) 
\notag  \\
&\quad
-h(\mbt{X} \setminus  \cup^k_{i=1}\mbt{Y}_{\sigma(i)} |   X_{\sigma(1)} +N_{\sigma(1)}, \ldots,  X_{\sigma(k)} +N_{\sigma(k)}, \mbt{Y}_{\sigma(1)},  \mbt{Y}_{\sigma(1)}, \ldots, \mbt{Y}_{\sigma(k)}).
\end{align}
Note that this last inequality is an equality
when $V = \emptyset$ and $U_{Y_\sigma(i)} = \wh{X}_{\sigma(i)} = X_{\sigma(i)} +N_{\sigma(i)}$, implying that $\bar{R}_\sigma = {R}_\sigma$.  Hence, 
\begin{align*}
R^m_{lb}(\mb{D}+ \epsilon\mb{1}) = \max_{\sigma} R_{\sigma}.
\end{align*}
From (\ref{app:lemma:strict_ineq:ineq1}), we know that $\max_{\sigma} R_{\sigma}$ is equal to $c \frac{1}{2}\log\frac{1}{D +\epsilon}$, where $c$ is an integer. Now we find $c$. When the permutation $\sigma(i) = 2i -1$, for $i \in [\frac{m-1}{2}]$ and  $\sigma(i) = 2(i -\frac{m-1}{2})$  for $i \in \{\frac{m-1}{2}, \ldots m\}$, we get $R_{\sigma} = \frac{m-1}{2} \frac{1}{2}\log\frac{1}{D +\epsilon}$, implying $c \ge \frac{m-1}{2}$. Also, from  Theorem \ref{thm:odd_cycle_GI}, we know that  $c \le \frac{m}{2}$. Hence $c = \frac{m-1}{2}$. Then we have 

\begin{align}
R^m_{lb}(\mb{D}+ \epsilon\mb{1}) = \frac{m-1}{2} \frac{1}{2}\log\frac{1}{D +\epsilon}.
\end{align}
Taking $\epsilon \rightarrow 0$ on both sides gives the result.
\end{proof}

\section{}
\label{app:typicality}
We first give the definition of $\epsilon$-letter typical sequences \cite{kramer} and then reference results that are useful to prove Theorem \ref{thm:gen_ach}.

\begin{definition} Let $\epsilon> 0$ be given. $x^n \in \mathcal{X}^n$ is called an \textit{$\epsilon$-letter typical sequence} with respect to $p_X$ if 
\begin{align*}
\left|\frac{1}{n}N(a|x^n) - p_X(a)\right| \le p_X(a)\epsilon, \mbox{ for all } a \in \mathcal{X},  
\end{align*}
where $N(a|x^n)$ denotes the number of times the symbol $a$ occurs in $x^n$.
Also $\mathcal{T}^{(n)}_{\epsilon}(p_X)$ denotes the set of all $\epsilon$-letter typical sequences with respect to $p_X$.
\end{definition}

\begin{definition} Let $\epsilon> 0$ be given. $(x^n, y^n) \in \mathcal{X}^n \times \mathcal{Y}^n$ is called a \emph{ jointly  typical sequence} with respect to $p_{XY}$ if 
\begin{align*}
\left|\frac{1}{n}N(a, b|x^n,y^n) - p_{XY}(a,b)\right| \le p_{XY}(a,b)\epsilon, \mbox{ for all } (a,b) \in \mathcal{X} \times \mathcal{Y}.
\end{align*}
Also $\mathcal{T}^{(n)}_{\epsilon}(p_{XY})$ denotes the set of all jointly typical sequences with respect to $p_{XY}$.
\end{definition}

\begin{definition} Let $\epsilon> 0$ be given. The set of\emph{ conditionally typical sequences}, $\mathcal{T}^{(n)}_{\epsilon}(p_{XY}|x^n)$, is defined as 
\begin{align*}
\mathcal{T}^{(n)}_{\epsilon}(p_{XY}|x^n) = \{ y^n | (x^n,y^n) \in \mathcal{T}^{(n)}_{\epsilon}(p_{XY}) \}.
\end{align*}
\end{definition}

\begin{lemma}\cite[Theorem 1.1]{kramer}
\label{thm:typical1}
Let $0 < \epsilon \le \mu_X$ where $\mu_X = \min_{x \in support(p_X)}p(x)$ and $X^n \in \mathcal{X}^n$ is drawn i.i.d.\ with respect to $p_X$. Then 
\begin{align*}
1 - \delta_\epsilon(n) \le  \Pr [X^n \in  \mathcal{T}^{(n)}_\epsilon(p_X)] \le 1,
\end{align*}
where $\delta_{\epsilon}(n) = 2|\mathcal{X} | e^{-n\epsilon^2\mu_X}$.
\end{lemma}

\begin{lemma}
\label{thm:typical_joint}\cite[Theorem 1.3]{kramer}
Let $0 < \epsilon_1 < \epsilon_2 \le \mu_{XY}$ where $\mu_{XY} = \min_{(x,y) \in support(p_{XY})}p(x,y)$ and $Y^n \in \mathcal{Y}^n$  drawn i.i.d.\ with respect to $p_Y$. If $x^n \in \mathcal{T}^{(n)}_{\epsilon_1}(p_X)$ then 
\begin{align*}
\left(1 - \delta_{\epsilon_1, \epsilon_2}(n)\right) 2^{-n\left(I(X;Y) + 2\epsilon_2H(Y)\right)} 
\leq \Pr\left[ Y^n \in \mathcal{T}^{(n)}_{\epsilon_2}\left(p_{XY}\mid x^n\right) \right] \leq 2^{-n\left(I(X;Y)-2 \epsilon_2 H(Y) \right)},
\end{align*}
where $\delta_{\epsilon_1, \epsilon_2}(n) = 2|\mathcal{X}||\mathcal{Y}| \cdot e^{-n\frac{(\epsilon_2 -
    \epsilon_1)^2}{1+\epsilon_1}\mu_{XY}}$.
\end{lemma}

\begin{corollary}
\label{corr:typical_joint}
Let $0 < \epsilon_1 < \epsilon_2 \le \mu_{XYZ}$ where $\mu_{XYZ} = \min_{(x,y,z) \in support(p_{XYZ})}p(x,y,z)$. $Y^n \in \mathcal{Y}^n$ is drawn i.i.d.\ with respect to $p_Y$  and $Z^n \in \mathcal{Z}^n$  is drawn i.i.d.\ with respect to $p_Z$. If $x^n \in \mathcal{T}^n_{\epsilon_1}(p_X)$ then 
\begin{align*}
 \Pr\left[ (Y^n,Z^n) \in \mathcal{T}^{(n)}_{\epsilon_2}\left(p_{XYZ}\mid x^n\right) \right] \leq 2^{-n\left((H(Y) + H(Z) - H(Y,Z|X))-2 \epsilon_2 (H(Y) + H(Z)) \right)}.
\end{align*}
\end{corollary}
\begin{proof}
\begin{align*}
\Pr\left[ (Y^n,Z^n) \in \mathcal{T}^{(n)}_{\epsilon_2}\left(p_{XYZ}\mid x^n\right) \right] &= \sum_{(y^n, z^n) \in \mathcal{T}^{(n)}_{\epsilon_2}\left(p_{XYZ}\mid x^n\right)} p^n_Y(y^n)p^n_Z(z^n)
\\
& \le 2^{- n(1 - \epsilon_2)H(Y)}2^{- n(1- \epsilon_2)H(Z)}|\mathcal{T}^{(n)}_{\epsilon_2}\left(p_{XYZ}\mid x^n\right)|, \mbox{ by \cite[Theorem 1.1]{kramer} }
\\
& \le 2^{- n(1 - \epsilon_2)H(Y)}2^{- n(1- \epsilon_2)H(Z)}2^{nH(Y, Z |X)(1+\epsilon_2)}, \mbox{ by \cite[Theorem 1.2]{kramer} }
\\
& \le  2^{- n\left( (H(Y) + H(Z) - H(Y,Z|X)) - 2\epsilon_2(H(Y)+H(Z))\right)}.
\end{align*}
\end{proof}
\begin{lemma}\cite[Markov Lemma]{kramer}
\label{thm:typical_markov}
Let $0 < \epsilon_1 < \epsilon_2 \le \mu_{XYZ}$ where $\mu_{XYZ} = \min_{(x,y,z) \in support(p_{XYZ})}p(x,y,z)$ and $(X^n, Y^n, Z^n )$ is  drawn i.i.d.\ with respect to $p_{XYZ}$ such that $X \lra Y \lra Z$. If $(x^n,y^n) \in \mathcal{T}^n_{\epsilon_1}(p_{XY})$ then 
\begin{align*}
 \Pr\left[ Z^n \in \mathcal{T}^{(n)}_{\epsilon_2}\left(p_{XYZ}\mid x^n, y^n \right)| Y^n =y^n \right] 
&= \Pr\left[ Z^n \in \mathcal{T}^{(n)}_{\epsilon_2}\left(p_{XYZ}\mid x^n, y^n\right)| Y^n =y^n, X^n =x^n \right] 
\\
& \ge 1 - \delta_{\epsilon_1, \epsilon_2}(n)
\end{align*}
where $\delta_{\epsilon_1, \epsilon_2}(n) = 2|\mathcal{X}||\mathcal{Y}||\mathcal{Z}| \cdot e^{-n\frac{(\epsilon_2 -
    \epsilon_1)^2}{1+\epsilon_1}\mu_{XYZ}}$.
\end{lemma}

\begin{lemma}
\label{lemma:bin_indp}
Let $A$, $B$ and $C$ denote the events 
\begin{align*}
&\{ \exists \mb{\wt{k}}_{\set{D}'_l}  \mbox{ such that } \mb{\wt{k}}_{\set{D}'_l} \neq \mb{k}_{\set{D}'_l}, \wt{k}_{\set{D}'_l} = k_{\set{D}'_l},
 U_{\set{D}'_l}(\mb{\wt{k}}_{\set{D}'_l}) \in  \mathcal{T}^{(n)}_{\epsilon_{l^*+2}}(p| U_{\set{D}_l\setminus \set{D}'_l}(\mb{k}_{\set{D}_l\setminus \set{D}'_l}),Y^n_l)\} \mbox{ and }
 \\
 &  \{\exists \mb{\wt{k}}_{\set{D}'_l}  \mbox{ such that } \wt{k}_{\set{D}'_l} = \mb{1},
 U_{\set{D}'_l}(\mb{\wt{k}}_{\set{D}'_l}) \in  \mathcal{T}^{(n)}_{\epsilon_{l^*+2}}(p| U_{\set{D}_l\setminus \set{D}'_l}(\mb{k}_{\set{D}_l\setminus \set{D}'_l}),Y^n_l) \}
 \\
 &
 \{ (U_{\set{D}_l\setminus \set{D}'_l}(\mb{k}_{\set{D}_l\setminus \set{D}'_l}),Y^n_l )\in  \mathcal{T}^{(n)}_{\epsilon_{l^*+1}}(p) \}
\end{align*}
respectively. Then 
\begin{align*}
&\Pr\left( A |C\right) 
 \le
  \Pr\left( B|C\right).
\end{align*}
\end{lemma}
\begin{proof}
The proof follows the steps in \cite[Lemma 11.1]{ElGamal}. 
We start with showing that for a particular set of bin indices $b_{\set{D}'_l}$,
\begin{align}
\label{ineq:lemma_bin_indp}
&\Pr\left(  A | C, k_{\set{D}'_l} = b_{\set{D}'_l} \mbox{ is chosen at the encoder }\right) 
 \le
  \Pr\left(  B | C, k_{\set{D}'_l} = b_{\set{D}'_l}  \mbox{ is chosen at the encoder }\right).
 \end{align}
 We can write
 \begin{align}
 &\Pr\left(  A | C, k_{\set{D}'_l} = b_{\set{D}'_l} \mbox{ is chosen at the encoder}\right)
\notag  \\
 &= \sum_{  b'_{\set{D}'_l} }p(b'_{\set{D}'_l}| b_{\set{D}'_l} ) 
 \Pr\bigg(  \exists \mb{\wt{k}}_{\set{D}'_l}  \mbox{ such that } \wt{k}_{\set{D}'_l} = b_{\set{D}'_l}, \wt{k}'_{\set{D}'_l} \neq b'_{\set{D}'_l}, 
 U_{\set{D}'_l}(\mb{\wt{k}}_{\set{D}'_l}) \in  \mathcal{T}^{(n)}_{\epsilon_{l^*+2}}(p|U_{\set{D}_l\setminus \set{D}'_l}(\mb{k}_{\set{D}_l\setminus \set{D}'_l}),Y^n_l) \bigg| 
\notag \\ 
 & \quad  \quad  \quad  \quad  \quad  \quad  \quad  \quad  \quad \quad
  \quad  \quad  \quad  \quad
C,  \mb{k}_{\set{D}'_l} =( b_{\set{D}'_l}, \bar{b}'_{\set{D}'_l}) \mbox{ is chosen at the encoder} \bigg)
\notag  \\
  & \overset{a}{=} \sum_{ b'_{\set{D}'_l} }p(b'_{\set{D}'_l}  | b_{\set{D}'_l} ) \Pr\bigg(  \exists \mb{\wt{k}}_{\set{D}'_l} \mbox{ such that } \wt{k}_{\set{S}_j} = 1, \wt{k}'_{\set{S}_j} \in [2^{R'_{\set{S}_j}}-1] \mbox{ for all } \set{S}_j \in \set{D}'_l,
  \notag \\
   & \quad  \quad  \quad  \quad \quad \quad \quad \quad \quad \quad
 U_{\set{D}'_l}(\mb{\wt{k}}_{\set{D}'_l}) \in  \mathcal{T}^{(n)}_{\epsilon_{l^*+2}}(p|U_{\set{D}_l\setminus \set{D}'_l}(\mb{k}_{\set{D}_l\setminus \set{D}'_l}),Y^n_l) \bigg| 
C,  \mb{k}_{\set{D}'_l} =( b_{\set{D}'_l}, \bar{b}'_{\set{D}'_l}) \mbox{ is chosen at the encoder} \bigg)
\notag  \\
  & \overset{b}{\le} \sum_{ b'_{\set{D}'_l} }p(b'_{\set{D}'_l}  | b_{\set{D}'_l} ) \Pr\bigg(  \exists \mb{\wt{k}}_{\set{D}'_l} \mbox{ such that } \wt{k}_{\set{S}_j} = 1 \mbox{ for all } \set{S}_j \in \set{D}'_l,
 U_{\set{D}'_l}(\mb{\wt{k}}_{\set{D}'_l}) \in  \mathcal{T}^{(n)}_{\epsilon_{l^*+2}}(p|U_{\set{D}_l\setminus \set{D}'_l}(\mb{k}_{\set{D}_l\setminus \set{D}'_l}),Y^n_l) \bigg| 
\notag \\ 
 & \quad  \quad  \quad  \quad  \quad  \quad  \quad  \quad  \quad \quad
  \quad  \quad  \quad  \quad
 C, \mb{k}_{\set{D}'_l} =( b_{\set{D}'_l}, \bar{b}'_{\set{D}'_l}) \mbox{ is chosen at the encoder} \bigg)
\notag  \\
   & =  \Pr\bigg(  \exists \mb{\wt{k}}_{\set{D}'_l} \mbox{ such that } \wt{k}_{\set{S}_j} = 1 \mbox{ for all } \set{S}_j \in \set{D}'_l,
 U_{\set{D}'_l}(\mb{\wt{k}}_{\set{D}'_l}) \in  \mathcal{T}^{(n)}_{\epsilon_{l^*+2}}(p|U_{\set{D}_l\setminus \set{D}'_l}(\mb{k}_{\set{D}_l\setminus \set{D}'_l}),Y^n_l) \bigg| 
\notag \\ 
 & \quad  \quad  \quad  \quad  \quad  \quad  \quad  \quad  \quad \quad
  \quad  \quad  \quad  \quad
C,  k_{\set{D}'_l} = b_{\set{D}'_l} \mbox{ is chosen at the encoder} \bigg)
 \notag \\
  \label{ineq:last_in_app}
  & = \Pr(B| C, k_{\set{D}'_l} = b_{\set{D}'_l} \mbox{ is chosen at the encoder}), 
 \end{align}
 where
 \\
 a : follows because given any set of codeword indices $\mb{b}_{\set{D}'_l} = ( b_{\set{D}'_l}, \bar{b}'_{\set{D}'_l})$ and event $C$, for each $\set{S}_j \in \set{D}'_l$, any collection of $[2^{R'_{\set{S}_j}}-1]$, the number of codewords $u^n(\mb{k}_{\set{S}_j})$ whose index $\mb{k}_{\set{S}_j}$ is different from $\mb{b}_{\set{S}_j}$  has the same distribution.
 \\
 b: Each bin in codebook $\mathcal{C}^{\set{S}_j}$ has size $2^{R'_{\set{S}_j}}$.
 \\
 
 Multiplying both sides of (\ref{ineq:last_in_app}) with $p(b_{\set{D}'_l})$ and summing over all bin indices $b_{\set{D}'_l}$ concludes the proof.
\end{proof}

\end{appendices}
\bibliographystyle{IEEEtran}

\bibliography{IEEEabrv,references}

\end{document}